\documentclass{article}
\pdfoutput=1
\usepackage{amssymb}
\usepackage[preprint]{neurips_2021}

\newcommand{\supp}{Supplementary Material}

     \usepackage[preprint]{neurips_2021}

\newcommand{\blank}[1]{}
\usepackage[utf8]{inputenc} % allow utf-8 input
\usepackage[T1]{fontenc}    % use 8-bit T1 fonts
\usepackage{hyperref}       % hyperlinks
\usepackage{url}            % simple URL typesetting
\usepackage{booktabs}       % professional-quality tables
\usepackage{amsfonts}       % blackboard math symbols
\usepackage{nicefrac}       % compact symbols for 1/2, etc.
\usepackage{microtype}      % microtypography
\usepackage{xcolor}         % colors
\usepackage{complexity}
\usepackage{lipsum}
\usepackage{physics}
\usepackage{bbold}
\usepackage{qcircuit}
\usepackage{graphicx}
\usepackage{dsfont}
\usepackage{amsthm}
\usepackage{qcircuit}
\usepackage{caption}
\usepackage{subcaption}
\usepackage{soul}
\usepackage{thm-restate}
\usepackage{placeins}
\newtheorem{theorem}{Theorem}[section]
\newtheorem{corollary}[theorem]{Corollary}
\newtheorem{definition}[theorem]{Definition}
\newtheorem{lemma}[theorem]{Lemma}

\title{Quantum Generative Training Using R\'enyi Divergences}

\author{
  M\'aria Kieferov\'a \\
  Centre for Quantum Computation and Communication Technology,\\
Centre for Quantum Software and Information,\\
University of Technology Sydney,\\
  \texttt{maria.kieferova@uts.edu.au} \\
\And
  Carlos Ortiz Marrero \\
  Data Sciences \& Analytics Group/Department of Electrical \& Computer Engineering \\
  Pacific Northwest National Laboratory/North Carolina State University\\
  \texttt{carlos.ortizmarrero@pnnl.gov} \\
   \And
  Nathan Wiebe \\
  Department of Computer Science/Department of Physics/High Performance Computing Group\\
  University of Toronto/University of Washington/Pacific Northwest National Laboratory\\
  \texttt{nawiebe@cs.toronto.edu} \\
}

\begin{document}

\maketitle

\begin{abstract}
Quantum neural networks (QNNs) are a framework for creating quantum algorithms that promises to combine the speedups of quantum computation with the widespread successes of machine learning. A major challenge in QNN development is a  concentration of measure phenomenon known as a barren plateau that leads to exponentially small gradients for a range of QNNs models. In this work, we examine the assumptions that give rise to barren plateaus and show that an unbounded loss function can circumvent the existing no-go results. We propose a training algorithm that minimizes the maximal R\'enyi divergence of order two and present techniques for gradient computation. We compute the closed form of the gradients for Unitary QNNs and Quantum Boltzmann Machines and provide sufficient conditions for the absence of barren plateaus in these models. We demonstrate our approach in two use cases: thermal state learning and Hamiltonian learning. In our numerical experiments, we observed rapid convergence of our training loss function and frequently archived a $99\%$ average fidelity in fewer than $100$ epochs.

\end{abstract}

\section*{Introduction}
The belief that quantum neural networks (QNN) can provide richer models for datasets than classical models has led to a lot of optimism about quantum computing transforming machine learning and artificial intelligence. Besides a range of advances, research over the last decade uncovered a major challenge to QNNs in form of gradient decay. This phenomenon, commonly referred to as "barren plateau", emerges at the very beginning of learning for the vast majority of models.
Specifically, QNNs exhibit a concentration of measure leading to exponentially small gradients (in the number of neurons) of the most common training loss functions. Barren plateaus are know to emerge in deep networks~\cite{mcclean2018barren}, shallow networks with global loss functions~\cite{cerezo2020cost}, as well as noisy quantum circuits~\cite{wang2020noise}. Further, recent work has shown that entanglement between hidden and visible layers within a quantum deep neural network can cause exponentially small gradients~\cite{marrero2020entanglement}.  Thus, the performance of QNNs becomes no better than random guessing with probability that is exponentially close to one over the number of hidden neurons. Even more, a recent result showed that under common assumptions barren plateaus arise whenever the QNN is highly expressible~\cite{holmes2021connecting} and they cannot be alleviated using gradient-free methods~\cite{arrasmith2020effect} or by using higher-derivatives~\cite{cerezo2021higher}. These results paint a bleak picture of the future of quantum machine learning and finding an approach for scalable training of generic QNNs is a central problem in quantum machine learning.

Here we give a novel approach to QNN training that circumvents known barren plateau results by violating one of the key assumptions. We empirically show that our algorithm does not suffer from gradient decay in the form of either entanglement induced barren plateaus, or their more conventional brethren, and we give sufficient criteria for avoiding barren plateaus. 
%Our approach to solve the problem involves introducing a new step in the process which we call quantum generative pre-training.  
%The motivation for this approach comes from training Boltzmann Machines by minimizing the KL-divergence.
In training, we learn the weights of QNNs by minimizing a loss function that qualifies how the quantum states generated by the neural network differ from the data. Existing algorithms almost exclusively utilize a linear bounded operator as a loss function. This condition is quite reasonable because the loss function is typically estimated by measuring the expectation values of Hermitian operators.
Our training algorithm minimizes a maximal R\'enyi divergence (akin to the KL-divergence) which upper bounds the quantum analog of the KL-divergence between two quantum states. The choice of this loss function implies that the standard arguments for barren plateau theorems do not apply because the R\'enyi-divergences  experience a logarithmic divergence when the two states are nearly orthogonal.   This causes the gradients of the divergence between nearly orthogonal quantum states to be large and thereby provides a workaround for all known barren plateau results. 

In our work, we derive a closed-form expression for the gradients and gave sufficient conditions for avoiding barren plateaus. Moreover, we observed an {\bf absence of gradient decay or barren plateaus} on all the learning tasks we performed. We implemented all our tools and experiments in Python, relying heavily on the QuTiP library \cite{johansson2012qutip}. Our code is available at \texttt{https://github.com/pnnl/renyiqnets}.

%In \cite{nadal2011statistical} it was shown that

\section{Preliminaries}
QNNs are natural generalizations of classical neural networks.  While there are many ways of thinking about the correspondence between ordinary (classical) neural networks and QNNs, we take the following correspondence in this work.  For simplicity, let us assume that we have a machine learning task where the training data is represented by $n_v$ binary features.  The training set can then be thought of as a probability distribution $P_{\rm train} \in \mathbb{R}^{2^{n_v}}$ such that $\sum_{\vec{x}} P(\vec{x}) =1$.  Here we generalize this problem to consider a training set that is a quantum distribution over the $n_v$ quantum bits, $\rho_{\rm train} \in \mathbb{C}^{2^{n_v} \times 2^{n_v}}$ such that ${\rm Tr}(\rho_{\rm train}) =1$ and $\rho_{\rm train}$ is positive semi-definite.  Using this formulation, classical data is a special case of quantum data where the state operator $\rho_{\rm train}$, also called a density matrix, is diagonal.

From this perspective, we can  think of a neural network as a parameterizable stochastic map such that if $n_h$ represents a number of hidden bits for a binary neural network and $w$ is a matrix of weights then $f: (w) \mapsto P_{v,h} \in \mathbb{C}^{2^{n_v+n_h}}$ such that $P_{v,h}(w)$ is a joint probability distribution on the visible and hidden units of the neural network.  Correspondingly, we take a QNN to be a quantum channel $\Lambda(w)$ such that $\Lambda: (w,\ket{0}\!\bra{0}) \mapsto \sigma \in \mathbb{C}^{2^{n_v+n_h} \times 2^{n_v + n_h}}$ such that ${\rm Tr}(\sigma)=1$ and $\sigma \succeq 0$. 
The QNN again reduces to the case of a classical neural network in the case where  $\sigma$ is diagonal for all $w$, and hence can be interpreted as a probability distribution for all weights $w$.  %Note here that we can take both $f$ and $\Lambda$ to be linear functions of their respective distributions $Q_v$ and $\sigma_v$ since linearity on distributions does not imply either operation invokes a linear transformation on their underlying quantum or classical bits.

Lastly we consider optimizing the training loss function, which we can imagine to be a functional of the density operator $\sigma_v$. 
Specifically, if we denote the loss function to be $\mathcal{L}: \mathbb{C}^{2^{n_v + n_h}\times 2^{n_v + n_h}} \mapsto \mathbb{R}$ then the goal of training a QNN is to find
\begin{equation}
    w^* = {\rm argmin}_w\mathcal{L}(\Lambda(w,\rho_0)).
\end{equation}
A common loss function for supervised training is the prediction loss given by $\mathcal{L}(\sigma_v ) =  {\rm Tr} (\sigma_v \sum_{x \in S_{\rm train}} \ket{x}\!\bra{x} \otimes \ket{\ell(x)}\!{\bra{\ell(x)}} \otimes I)$, where $S_{\rm train}$ is the training dataset and $\ell(x)$ is the label function for the training data.  For generative tasks, the quantum loss function between the training distribution $\rho_{\rm train}$ and the generated distribution on the visible units $\sigma_v$ is often chosen as the quantum relative entropy $\mathcal{L}(\rho) = S(\rho_{\rm train}|\sigma_v) = {\rm Tr}(\rho_{\rm train} \log (\rho_{\rm train})) - {\rm Tr}(\rho_{\rm train} \log (\sigma_v)$, which is the quantum analogue of the KL divergence between two distributions and in turn is a central concept for our notion about how to generatively train QNNs. Occasionally, only a subset of the qubits called "visible" qubits would serve as an output of the QNN and the remaining, "hidden" qubits would be treated as ancillary. That is, the output would be $\sigma_v = {\rm Tr}_h[\sigma]$ where we traced out over hidden qubits.

\subsection{Unitary Quantum Neural Networks}
The first model we consider is the unitary QNN, which consists of a sequence of parameterized operator exponentials.  This approach is arguably the most popular for implementing feed-forward neural networks; however, it is known to suffer problems from barren plateaus.  Our aim is to show that we can switch to a loss function that is appropriate for generative training.  Specifically, we will assume, without loss of generality, a set of operators $\{H_1,\ldots,H_N\}$ such that each $H_j$ is both Unitary and Hermitian and that the state yielded by the QNN is in $\mathbb{C}^{2^{n_v+n_h} \times 2^{n_v+n_h}}$ where $n_v$ and $n_h$ are the number of visible and hidden neurons, respectively, and
\begin{equation}
    \sigma(\theta):= \prod_{j=1}^N e^{-iH_j \theta_j} \ket{0}\!\bra{0} \prod_{j=N}^1 e^{iH_j \theta_j} \label{state_prep}
\end{equation}
as depicted in Fig.~\ref{fig:uqnn}.  Here the trace over the hidden subsystem corresponds to disregarding the values of the neurons in direct analogy to the way hidden layers are disregarded in classical neural networks~\cite{bishop2006pattern}. This approach has been used successfully in a host of quantum machine learning results~\cite{farhi2018classification,schuld2020circuit} and similar ansatzes have been used in quantum chemistry simulation~\cite{evangelista2019exact}.

\begin{figure}
    \centering
    \begin{subfigure}[b]{0.6\textwidth}
    \tiny{
    $$\Qcircuit @C=0.8em @R=0.8em {
\lstick{\ket{0}_v}&\qw&\multigate{1}{U(\theta_1)}&\qw&\multigate{1}{U(\theta_8)}&\qw&\dots&&\qw                &\multigate{1}{U(\theta_{n-3})}&\qw&\meter\\
\lstick{\ket{0}_v}&\qw&\ghost{U(\theta_1)}&\multigate{1}{U(\theta_5)}&\ghost{U(\theta_8)}&\qw&\dots&&\qw       &\ghost{U(\theta_{n-3})}&\qw&\meter\\
\lstick{\ket{0}_v}&\qw&\multigate{1}{U(\theta_2)}&\ghost{U(\theta_5)}&\multigate{1}{U(\theta_9)}&\qw&\dots&&\qw&\multigate{1}{U(\theta_{n-2})}&\qw&\meter\\
\lstick{\ket{0}_v}&\qw&\ghost{U(\theta_2)} &\multigate{1}{U(\theta_6)}&\ghost{U(\theta_9)} &\qw &\dots&&\qw    &\ghost{U(\theta_{n-2})}&\qw&\meter\\
\lstick{\ket{0}_v}&\qw&\multigate{1}{U(\theta_3)}&\ghost{U(\theta_6)}&\multigate{1}{U(\theta_{10})}&\qw&\dots&&\qw&\multigate{1}{U(\theta_{n-1})}&\qw&\meter\\
\lstick{\ket{0}_h}&\qw&\ghost{U(\theta_3)}&\multigate{1}{U(\theta_7)}&\ghost{U(\theta_{10})}&\qw&\dots&&\qw       &\ghost{U(\theta_{n-1})}&\qw\\
\lstick{\ket{0}_h}&\qw&\multigate{1}{U(\theta_4)}&\ghost{U(\theta_7)}&\multigate{1}{U(\theta_{11})}&\qw&\dots&&\qw&\multigate{1}{U(\theta_n)}&\qw\\
\lstick{\ket{0}_h}&\qw&\ghost{U(\theta_4)}&\qw&\ghost{U(\theta_{11})} &\qw&\dots&&\qw                         &\ghost{U(\theta_{n})}&\qw
}$$
}
        \caption{}
        \label{fig:uqnn}
    \end{subfigure} \hspace{10mm}
    \begin{subfigure}[b]{0.3\textwidth}
        \includegraphics[width=0.85\textwidth]{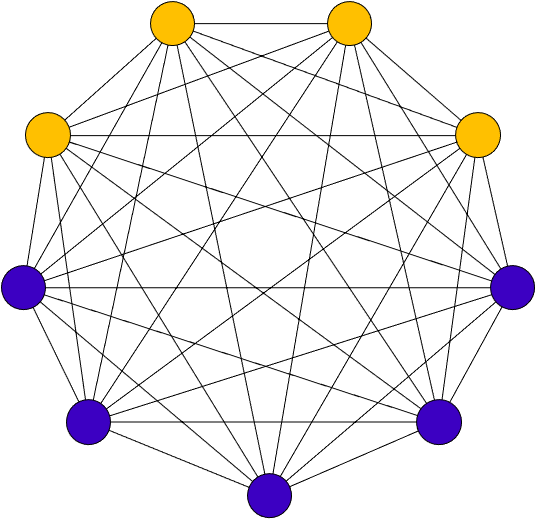}
        \caption{}
        \label{fig:qbm}
    \end{subfigure}
    \caption{Two types of QNNs. a) A Unitary QNN with $5$ visible units and $3$ hidden units. We take, for simplicity, unitaries of the form $U(\theta_j)=e^{-iH_j\theta_j}$. 
    b) A QBM with $6$ visible units and $4$ hidden units on a complete interaction graph.}\label{fig:architectures}
\end{figure}

\subsection{Quantum Boltzmann Machines}
Quantum Boltzman Machines (QBMs) are themselves a family of generative physics inspired QNNs.  They are specified by a Hamiltonian $H(\theta)\in \mathbb{C}^{2^{n_h +n_v}\times 2^{n_h+n_v}}$ such that 
\begin{equation}
\sigma(\theta) = \frac{{\rm Tr}_h \left(e^{-H(\theta)}\right)}{{\rm Tr}\left( e^{-H(\theta)}\right)}.
\end{equation}
This quantum state $e^{-H(\theta)}$ is chosen to be a Gibbs state because such states are the states of maximum von Neumann entropy constrained such that ${\rm Tr}(H(\theta) e^{-H(\theta)}) / {\rm Tr}(e^{-H(\theta)})$ is constant, which is desirable because doing so reduces the risk of overfitting. The Hamiltonian matrix $H$, in the most general case, can have any form.  
The QBM has the advantage of being simple to interpret in quantum settings, natural to implement in quantum annealers~\cite{amin2018quantum}, and easy to compute gradients for~\cite{kieferova2017tomography}. Moreover, it is particularly well suited for Hamiltonian learning problems that utilize thermal input states.

\subsection{Barren Plateaus}
A major challenge facing training QNNs was revealed in the seminal result of~\cite{mcclean2018barren}, wherein it was proven that with high probability the training objective functions most commonly used in quantum machine learning will be exponentially small for all but an exponentially small fraction of the values of $\theta$.  This problem is known as the barren plateau problem in the optimization landscape and we summarize it below.

\begin{definition}[Barren Plateau~\cite{volkoff2021large}]\label{thm:barren}
The cost function $M(\theta)$ exhibits a barren plateau with
respect to $\theta_j$ if it is continuously differentiable on a compact subset $A \subset \Omega $ and if for every $\epsilon > 0$, there exists $0 < b < 1$ such that $P_A(\left|\frac{\partial M}{\partial \theta_j}\right|) \geq \epsilon \in O\left( b^{n_v + n_h}  \right)$, where $P_A$ is the probability measure on $A$ induced by $P$, $\Omega$ is a compact set of parameters equipped with a probability density $P$, and $n_v$, $n_h$ are the numbers of visible and hidden units respectively.
\end{definition}

Barren plateaus are exhibited by a range of QNNs such as randomly initialized poly-depth circuits~\cite{mcclean2018barren}, shallow networks with a global cost function~\cite{cerezo2020cost}, and noisy quantum circuits~\cite{wang2020noise}.
As practical QNNs will likely need to contain thousands of neurons, this result shows that almost all such networks will have gradients that are negligibly small and cannot be trained efficiently even on a quantum computer. Moreover, recent results by Ortiz et al~\cite{marrero2020entanglement} have shown that the barren plateau problem can be even more significant for deep QNNs.  Specifically, entanglement between the visible and hidden layers in a QNN can cause the gradient to become exponentially small.  Specifically, they show that with probability at least $1-2^{-\Omega(n_h-n_v)}$ that the gradient will be in $O(\|\widehat{M}\| \sqrt{2^{n_v-n_h}})$ if the states satisfy a volume law for entanglement (which is typical for random quantum states) where $M(\theta)={\rm Tr}\left(\widehat{M}\sigma(\theta)\right)$ and $\widehat{M}$ is a Hermitian operator.  This shows that even if the typicality arguments of the McClean et al.~\cite{mcclean2018barren} result do not hold, then entanglement between the hidden and visible layers can still destroy the ability to train the network with respect to the bounded loss function.  Additional results show that even noise can induce a similar form of a barren plateau and thus any local optimization cannot be expected to be efficient in general for QNNs that satisfy these requirements. There are several approaches to overcoming barren plateaus that are either based on empirical evidence~\cite{grant2019initialization, skolik2021layerwise} or constrained to a specific architecture~\cite{pesah2020absence, cerezo2020cost, sharma2020trainability}  but there is no generic way of training that is guaranteed to avoid these no-go results.

The central point of this work is to show that such considerations do not appear when one transitions to unbounded loss functions, such as the quantum relative entropy, which are used in generative training.  We demonstrate how to perform this form of generative training for  Unitary QNNs and QBM and demonstrate the absence of barren plateaus.  Even if one is only interested in the performance of the quantum algorithm for discriminative tasks, our work suggests that there might be benefits in the form of initializing a QNN with a round of generative pre-training to ensure that the system does not get lost in a barren plateau.

\section{Generative training using R\'eyni entropy}
Perhaps the earliest work on the use of unitary QNNs is given in~\cite{romero2017quantum}.  The idea of this work was to construct a unitary QNN and then variationally optimize the overlap between the output state and a target state through a method such as the swap test.  The objective function for such a learning task is~$\rm{Tr}(\rho \sigma)$, where $\sigma$ is the output of the QNN and $\rho$ is the quantum training distribution (which is, without loss of generality, a mixed state).  This, unfortunately, runs afoul of the barren plateau problem raised in Definition~\ref{thm:barren}.

In our work, we suggest the use of a generalization of quantum relative entropy known as quantum R\'enyi divergence or "sandwiched" R\'enyi relative entropy ~\cite{wilde2014strong, muller2013quantum}. 
For two quantum states $\rho$ and $\sigma$, the quantum R\'enyi divergence $D_\alpha$ takes the form
$D_\alpha(\rho\|\sigma) = \frac{1}{\alpha-1} \log{\left[{\rm Tr} \left( \sigma^{\frac{1-\alpha}{2\alpha}}\rho \sigma^{\frac{1-\alpha}{2\alpha}} \right)^{\alpha} \right] }$ for $\alpha\in \left[0,\infty\right) \backslash \{1\}$. Quantum R\'enyi divergence inherits many of the mathematical properties of the R\'enyi divergence and in the case where $\alpha\rightarrow 1$ it reduces to the quantum relative entropy.
One can additionally define an upper-bound on the quantum R\'enyi divergence
\begin{equation}
  D_{\alpha}(\rho\|\sigma) \leq \widetilde{D}_\alpha (\rho\|\sigma) = \frac{1}{\alpha-1} \log{{\rm Tr} \left( \sigma^{\frac{1}{2}} \left( \sigma^{\frac{-1}{2}} \rho \sigma^{\frac{-1}{2}} \right)^{\alpha}\sigma^{\frac{1}{2}} \right)}
\end{equation}
where $\widetilde{D}_\alpha$ defines the maximal  R\'enyi divergence. For the purposes of training QNNs, we focus on $\alpha=2$, proposed by Petzl~\cite{petz1986quasi},
\begin{equation}
    \widetilde{D}_2(\rho\|\sigma(\theta)) = \log\left({\rm Tr}\left(\rho^2 \sigma^{-1}\right) \right).
\end{equation}
Here $\rho$ is the training data state and $\sigma(\theta)$  corresponds to the output of the QNN  as a function of the parameters $\theta$. 
The main argument for using $\widetilde{D}_2(\rho\|\sigma)$ as a loss function is that it upper-bounds the quantum relative entropy and its gradients are considerably simpler than that of the ordinary R\'enyi divergence and quantum relative entropy.  

We also considered the divergence with reversed arguments,
\begin{equation}
    \widetilde{D}_{2}(\sigma(\theta)|\rho) = \log\left({\rm Tr}\left(\sigma^2 \rho^{-1}\right) \right).
\end{equation}
Note that in general, $\widetilde{D}_{2}(\rho|\sigma(\theta)) \neq \widetilde{D}_{2}(\sigma(\theta)|\rho)$. However, if both $\rho$ and $\sigma(\theta)$ are full rank, $\widetilde{D}_{2}(\rho|\sigma(\theta)) = \widetilde{D}_{2}(\sigma(\theta)|\rho) = 0$ if and only if $\rho=\sigma$~\cite{renyi}.

Furthermore, these divergences have gradients that are straight forward to compute.  We give these gradients below and provide proof of their validity in the supplemental material.

\begin{restatable}[Unitary quantum neural network $\tilde{D}_2$  gradient, proof in \supp~\ref{sec:proof_qnn}]
{thm}{unitarygradient} \label{thm:QNN} 
 Let $\sigma_v = Tr_h\left[ \Pi_{j=1}^N e^{-iH_j\theta_j} \ket{0}\!\!\bra{0} \Pi_{j=N}^1 e^{iH_j\theta_j} \right]$ be the reduced state on the visible subsystem $\mathcal{H}_v$. Then the gradients of the maximal R\'enyi divergence between $\rho$ and $\sigma_v$ takes the form
\begin{equation}
    \partial_{\theta} \widetilde{D}_2(\rho\|\sigma_v(\theta)) = \frac{i {\rm Tr}\left(\rho^2 \sigma_v^{-1} {\rm Tr}_h(
    [\widetilde{H}_k, \sigma]
    )\sigma_v^{-1}\right)}{{\rm Tr}\left(\rho^2 \sigma_v^{-1}\right)}. \label{gradient_uni_rho}
\end{equation}
where $\widetilde{H}_k = \prod_{j=1}^{k-1} e^{-iH_j \theta_j} H_k \prod_{j=k-1}^1 e^{iH_j \theta_j}$. Similarly, the gradient of the reverse divergence $ \widetilde{D}_2(\sigma_v|\rho)$ can be expressed as
\begin{equation}
    \partial_{\theta_k} \widetilde{D}_2(\sigma_v\|\rho) = \frac{-i{\rm Tr}\left(\left\{{\rm Tr_h}({[\widetilde{H}_k, \sigma]}),\sigma_v\right\} \rho^{-1}\right)}{{\rm Tr}\left(\sigma_v^2 \rho^{-1}\right)}. \label{gradient_uni_sigma}
\end{equation} 
\end{restatable}

\begin{restatable}[Quantum Boltzmann machine $\tilde{D}_2$ gradient, proof in \supp~\ref{sec:proof_qbm} ]{thm}{qbmgradient} \label{thm:boltz} 
\ \\ Let $\sigma_v(\theta)=Tr_h[e^{-H}/Tr[e^{-H}]]$ be a state on the visible units of a QBM and let ${\rm Ad}_H(A)= [H,A]$ be the adjoint endomorphism for some matrix $H$. The gradient of the objective function $\widetilde{D}_2(\rho\|\sigma_v(\theta)) $ are,

\begin{equation}
    \partial_{\theta} \widetilde{D}_2(\rho\|\sigma_v(\theta)) 
     =\sum_{p=0}^\infty\frac{ {\rm Tr}\left(\rho^2 \sigma_v^{-1}(\theta) ({\rm Tr_h}\left({\rm Ad}^p_{-H}(\partial_\theta H)e^{-H}\right)\sigma_v^{-1}(\theta)\right)}{{\rm Tr}\left(\rho^2 \sigma_v^{-1}\right){\rm Tr}\left( e^{-H(\theta)}\right)(p+1)!}-\frac{{\rm Tr}\left(\left(\partial_{\theta} H(\theta)\right)e^{-H(\theta)} \right)}{{\rm Tr}\left( e^{-H(\theta)}\right)}
\end{equation}

and for the reverse divergence,

\begin{equation}
    \partial_{\theta} \widetilde{D}_2(\sigma_v(\theta)\|\rho) 
    = -\sum_{p=0}^\infty\frac{{\rm Tr}\left({\rm Ad}_{-H}^{p}(\partial_\theta H) e^{-H}(\{\sigma_v(\theta),\rho^{-1}\})\otimes I_h\right)}{{\rm Tr}\left(\sigma_v^2 \rho^{-1} \right){\rm Tr}\left( e^{-H(\theta)}\right)(p+1)!}+2\frac{{\rm Tr}\left(\left(\partial_{\theta} H(\theta)\right)e^{-H(\theta)} \right)}{{\rm Tr}\left( e^{-H(\theta)}\right)}
    \label{gradient_bm}
\end{equation}
\end{restatable}

Our quantum algorithm for evaluating the terms in these gradient expressions is given below.

\begin{restatable}[Extended Swap test, proof in \supp~\ref{sec:proof_swap}]{thm}{swapamard}\label{thm:swap_trick}
\ \\ Let $\rho_1,\ldots, \rho_n$ be $m$-qubit densitry operators and let $U_1,\ldots, U_n$ be unitary operations on $m$ qubits.  There exists a unitary operation $\chi$ consisting of two Hadamard gates, $n$ controlled $U_j$ operations and $nm$ controlled swap gates such that for any state of the form $\rho_1 \otimes\cdots\otimes \rho_n$,
\begin{equation}
    {\rm Tr}\left(\left(\ket{0}\!\bra{0} \otimes I\right)\chi \left( \ket{0}\!\bra{0} \otimes \rho_1 \otimes \dots \otimes \rho_n \right) \chi^{\dagger} \right)= \frac{1+ Re[{\rm Tr} (\prod_{i=1}^n U_i \rho_i)]}{2}.
\end{equation}

\end{restatable}

One can estimate the above gradient by computing the numerator and the denominator separately and then propagating the error in both estimates. For unitary QNNs and QBMs, one can estimate individual terms in the gradient by sampling utilizing Theorem~\ref{thm:swap_trick} and the linearity of the trace. We focus on the reverse divergence here because it is often easier to implement $\rho^{-1}$ than $\sigma^{-1}$ which changes throughout the gradient descent.  

In order to prepare such states, we assume access to an  unitary quantum channel preparing a purification of $\rho^{-1}$ through an oracle such that 
\begin{equation}
    {\rm Tr}_{\rm anc}(\mathcal{O}_{\rho^{-1}} \ket{0}\otimes \ket{0}_{\rm anc}) = \rho^{-1}/\rm{Tr}(\rho^{-1}).
\end{equation}
 We also assume oracular access to the elementary Hamiltonians terms $H_j$ that we require to be both hermitian and unitary.  In particular, let us assume a unitary and invertible quantum oracle
 \begin{equation}
   \mathcal{O}_H:\ket{j}\ket{\psi} \mapsto \ket{j} H_j \ket{\psi}  
\end{equation}
for any state $\ket{\psi}$.  Lastly, assume an oracle 
 \begin{equation}
 \mathcal{O}_{exp}(t)\ket{j}\ket{\psi} = \ket{j}e^{-iH_j t}\ket{\psi}  
 \end{equation}
where $H_j$s are the elementary Hamiltonians used for constructing a unitary QNN. %In the most common scenario, each $H_j$ would be a tensor product $O(1)$ Pauli operators. If $H_j$ is $k$-local, the cost of applying $H_j$ and $e^{-i H_j \theta}$ would be $O(k)$~\cite{nielsen2002quantum}.  
We further assume that each oracle can be made {singly} controlled at unit cost.  Under these assumptions, we state the complexity of the gradient evaluation below,

\begin{restatable}[Complexity of computing the gradient of reverse divergence for unitary QNNs]{thm}{generic}\label{thm:gradient_generic}
Let $\rho$ be the training data state and $\sigma_v = {\rm Tr}_h \left[ \prod_{j=1}^{N} e^{-iH_j\theta_j} \ket{0}\!\!\bra{0} \prod_{j=N}^{1} e^{iH_j\theta_j} \right]$ be the output of a unitary QNN on visible units.  For any $\epsilon>0$, the number of single-qubit gates, Toffoli gates and queries to the oracles $\mathcal{O}_{\rho^{-1}}$, $\mathcal{O}_U$ and $\mathcal{O}_{exp}$ needed to compute an estimate $\mathcal{E}$ such that $|\mathcal{E} - \partial_{\theta_k} D(\sigma||\rho)| \le \epsilon \le \frac{13}{2{\rm Tr}(\sigma_v^2 \rho^{-1})}$ with probability greater than $2/3$ is in 
\begin{equation*}
    O\left( \frac{N+n_v} {\epsilon   {\rm Tr}(\sigma_v^2 \rho^{-1} /{\rm Tr}(\rho^{-1})) } 
    \right).
\end{equation*}
\end{restatable}

\section{Absence of Barren Plateaus for R\'enyi 2-divergence}
Above we derived the gradients for both the forward and reverse gradients of the R\'enyi 2-divergence,  but these gradients still might suffer from exponentially small gradients even if they do not satisfy the assumptions of existing barren plateau theorems. Below we provide sufficient conditions for the absence of a barren plateau for this training loss functions.

\begin{lemma}[Proof in \supp~\ref{sec:plateau}]\label{lem:plateau}
Let $U$ be a unitary matrix drawn from a unitarily invariant measure, such as the Haar measure.  Further let $\lambda_i(\cdot)$ denote the $i^{\rm th}$ eigenvalue in a sorted list of eigenvalues of the matrix $(\cdot)$ with $\lambda_0(\cdot)$ being the smallest eigenvalue.  For an all-visible QNN with $n$ neurons with a training distribution $\rho \succ 0$ we then have that
\begin{enumerate}
\item $\mathbb{E}\left(\left(\partial_{\theta_k} \widetilde{D}_2(\rho\|U\sigma U^\dagger) \right)^2\right) \in \Omega\left( \frac{{\rm Tr}^2(\sigma^{-2} (\partial_{\theta_k} \sigma))}{2^{2n} {\rm Tr}^2(\sigma^{-1})}\right)\subseteq \Omega \left({2^{-2n}} \lambda_0(\sigma^{-1} (\partial_{\theta_k} \sigma)^2 \sigma^{-1} \right)$
\item $\mathbb{E}\left(\left(\partial_{\theta_k} \widetilde{D}_2(U\sigma U^\dagger\|\rho ) \right)^2\right)\in \Omega \left(\frac{{\rm Tr}^2(\sigma(\partial_{\theta_k}\sigma ))}{2^{2n}\|\sigma\|^4 }\right)\subseteq \Omega\left( 2^{-2n} \|\sigma\|^{-4}\lambda_0((\partial_{\theta_k} \sigma)^2) \right)$

\end{enumerate}
\end{lemma}
These conditions are rather benign, which suggests that under most circumstances barren plateau should not appear for generative training using $\widetilde{D}_2$.  To get an intuition for this, consider training a unitary QNN.  Let us consider the reverse divergence. Notice that a typical random density operators will have $\|\sigma\| \approx 1/2^n$, so we expect the gradients to be at least $2^n\lambda_0(\partial_{\theta_k} \sigma)$, which we expect to be on the order of $\|H_k\|\in \Theta(1)$.  A similar argument leads to the same conclusion for the reverse divergence if we note that the mean value of the minimum eigenvalue is in $2^{-3n}$~\cite{chen2010smallest} and thus the typical scale of the inverse of the density matrix is $2^{3n}/2^{n} = 2^{2n}$. From this scaling argument we then similarly expect the gradients to be on the order of $\|H_k\|$.  For these reasons, we believe that barren plateaus are expected to be absent in the optimization landscape.  Our expectations are further validated when considering small scale numerical experiments.

\section{Application: Learning thermal states}\label{sec:thermal}

We apply our R\'enyi divergence training routine to learn thermal states, i.e. $\rho = e^{-H}/{\rm Tr}(e^{-H})$. For a full description and complexity analysis, see \supp~\ref{sec:app_thermal_state}. The description of the state $\rho$ is in terms of its Hamiltonian that is given to us through an oracle in individual terms $U_l$ in its LCU decomposition $H=\sum_{l=1}^L \alpha_l U_l$. We also assume a second oracle $\mathcal{O}_{exp}(t)\ket{j}\ket{\psi} = \ket{j}e^{-iH_j t}\ket{\psi}$ where $H_j$ are the elementary Hamiltonians used for constructing a QNN. 

For thermal states, the gradient of the reverse R\'enyi divergence~\eqref{gradient_uni_sigma} gives
\begin{equation}
    \partial_{\theta_k} \widetilde{D}_2(\sigma_v\|\rho) = \frac{-i{\rm Tr}\left(\left\{{\rm Tr_h}({[\widetilde{H}_k, \sigma]}),\sigma_v\right\} e^{H}\right)}{{\rm Tr}\left(\sigma_v^2 e^{H}\right)}. \label{gradient_thermal}
\end{equation}
We propose two algorithms for estimating~\eqref{gradient_thermal} that extend~\ref{thm:gradient_generic}. For both of them, we estimate the numerator and then the denominator separately and return our estimate of the gradient to be the quotient of the two. Our first algorithm prioritizes low depth and circuit simplicity over asymptotic complexity. The second algorithm builds on the first but uses coherent operations instead of sampling to increase efficiency. This algorithm builds on our low-depth algorithm but uses amplitude estimation and LCU state preparation depicted in Fig.~\ref{fig:LCU}. The main component of both algorithms is the extended swap test. 

The coherent algorithm explicitly uses a unitary query model for the Hamiltonian of the following form:
For $H=\sum_{j=1}^L \alpha_j U_j$, define the operation $\text{Prepare}$ and $\text{Select}$ such that
\begin{align} 
    \text{Prepare}\ket{0} &= \frac{1}{\sqrt{s}} \sum_{k, l_1, \dots, l_k} \sqrt{\frac{\alpha_{l_1}\dots\alpha_{l_k}}{2^kk!}} \ket{k, l_1, \dots, l_k}\\
    \text{Select}\ket{j}\ket{\psi} &=\ket{j} U_j \ket{\psi},
\end{align}
for $s \le e^{\|\alpha\|_1/2} = e^{\sum_j |\alpha_j|/2}$.  The cost of performing computing the gradient for a thermal state learning problem using these operations as fundamental oracles instead of $\mathcal{O}_{\rho^{-1}}$ is given below.

 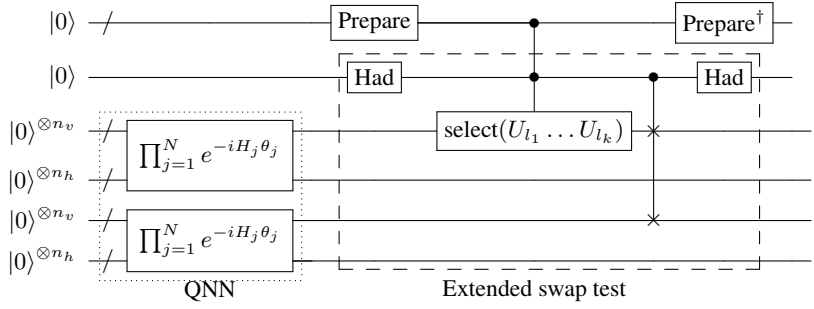
\begin{figure}[t]
 \centering
 \footnotesize{
 $$\Qcircuit @C=0.8em @R=0.8em {
 \lstick{\ket{0}}&{/}\qw&\qw&\qw&\gate{\text{Prepare}}&\ctrl{1}\qw&\qw&\gate{\text{Prepare}^{\dagger}}&\qw\\
 \lstick{\ket{0}}&\qw &\qw&\qw& \gate{\text{Had}}&\ctrl{1}&\ctrl{3}&\gate{\text{Had}}&\qw\\
\lstick{\ket{0}^{\otimes n_v}}&{/}\qw&\multigate{1}{\prod_{j=1}^{N} e^{-iH_j\theta_j}}&\qw&\qw&\gate{\text{select}(U_{l_1}\dots U_{l_k})}&\qswap&\qw&\qw&\qw\\
 \lstick{\ket{0}^{\otimes n_h}}&{/}\qw&\ghost{\prod_{j=1}^{N} e^{-iH_j\theta_j}}&\qw&\qw&\qw&\qw&\qw&\qw&\qw\\
 \lstick{\ket{0}^{\otimes n_v}}&{/}\qw&\multigate{1}{\prod_{j=1}^{N} e^{-iH_j \theta_j}}&\qw& \qw&\qw&\qswap&\qw&\qw&\qw\\
 \lstick{\ket{0}^{\otimes n_h}}&{/}\qw&\ghost{\prod_{j=1}^{N} e^{-iH_j \theta_j}}&\qw \qw&\qw&\qw&\qw&\qw&\qw&\qw  \\
 & & \mbox{QNN} & & &  \mbox{Extended swap test}
 \gategroup{2}{5}{6}{8}{.7em}{--} 
 \gategroup{3}{2}{6}{3}{.7em}{..}
 }$$}
 \caption{The circuit for initializing the QNN and extended swap utilizing LCU.  These operations will be repeated multiple times within amplitude estimation.}
 \label{fig:LCU}
 \end{figure}

\begin{restatable}[Coherent algorithm for thermal state learning]{thm}{thmdeep}\label{thm:deep_qnn}

Let $\rho = e^{-H}/{\rm Tr}(e^{-H})$ be a target distribution for $H=\sum_{l=1}^L \alpha_l U_l$ for unitary $U_l$ and $\sigma_v = {\rm Tr}_h \left[ \prod_{j=1}^{N} e^{-iH_j\theta_j} \ket{0}\!\!\bra{0} \prod_{j=N}^{1} e^{iH_j\theta_j} \right]$ be the output of a unitary QNN on visible units.  For any $\epsilon>0$, the number of gates needed to compute an estimate $\mathcal{E}$ such that $|\mathcal{E} - \partial_{\theta_k} \widetilde{D}_2(\sigma||\rho)| \le \epsilon\le \frac{13}{2{\rm Tr}(\sigma_v^2 \rho^{-1})}$ with probability greater than $2/3$ is in 
$$
         \widetilde{O}\left(\frac{N+\sqrt{\frac{{e^{\|\alpha\|_1} 2^n}}{{\rm Tr}(e^H)}}}{\epsilon {\rm Tr}(\sigma_v^2 e^H/{\rm Tr}(e^H))}\right),
$$
Similarly, the number of queries  to the oracles $\mathcal{O}_U$ and $\mathcal{O}_{exp}$, Prepare and Select needed  is in
$$
        \widetilde{O}\left(\frac{(L+n)\sqrt{\frac{{e^{\|\alpha\|_1} 2^n}}{{\rm Tr}(e^H)}}}{\epsilon {\rm Tr}(\sigma_v^2 e^H/{\rm Tr}(e^H))}\right),
$$
where the notation $\tilde{O}(\cdot)$ suppresses logarithmic factors.
\end{restatable}

 \begin{figure}[t]
  \centering%
  \subfloat[Loss]{%
    \includegraphics[width=0.48\textwidth]{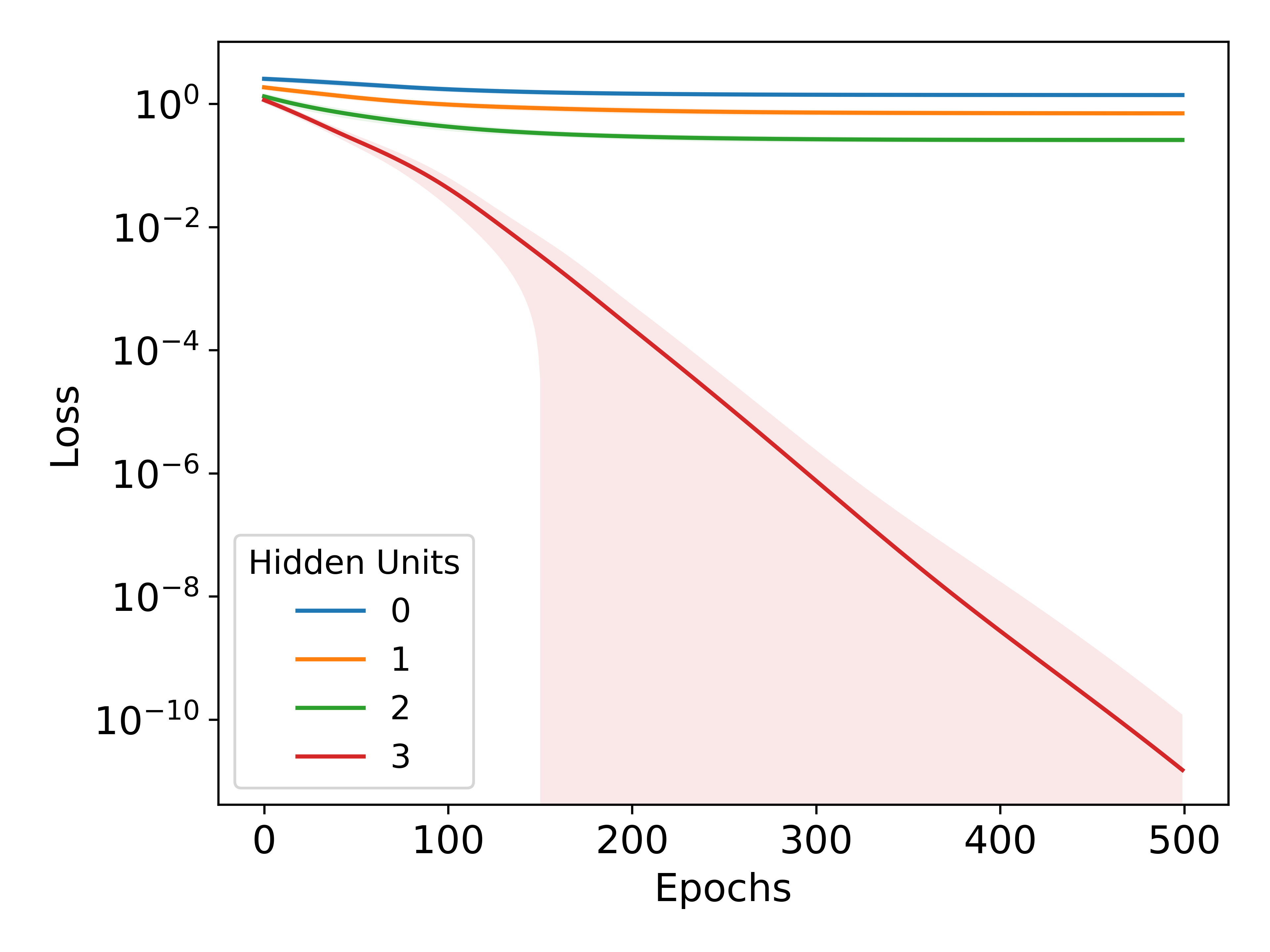}%
  }%
    \subfloat[Fidelity]{%
    \includegraphics[width=0.48\textwidth]{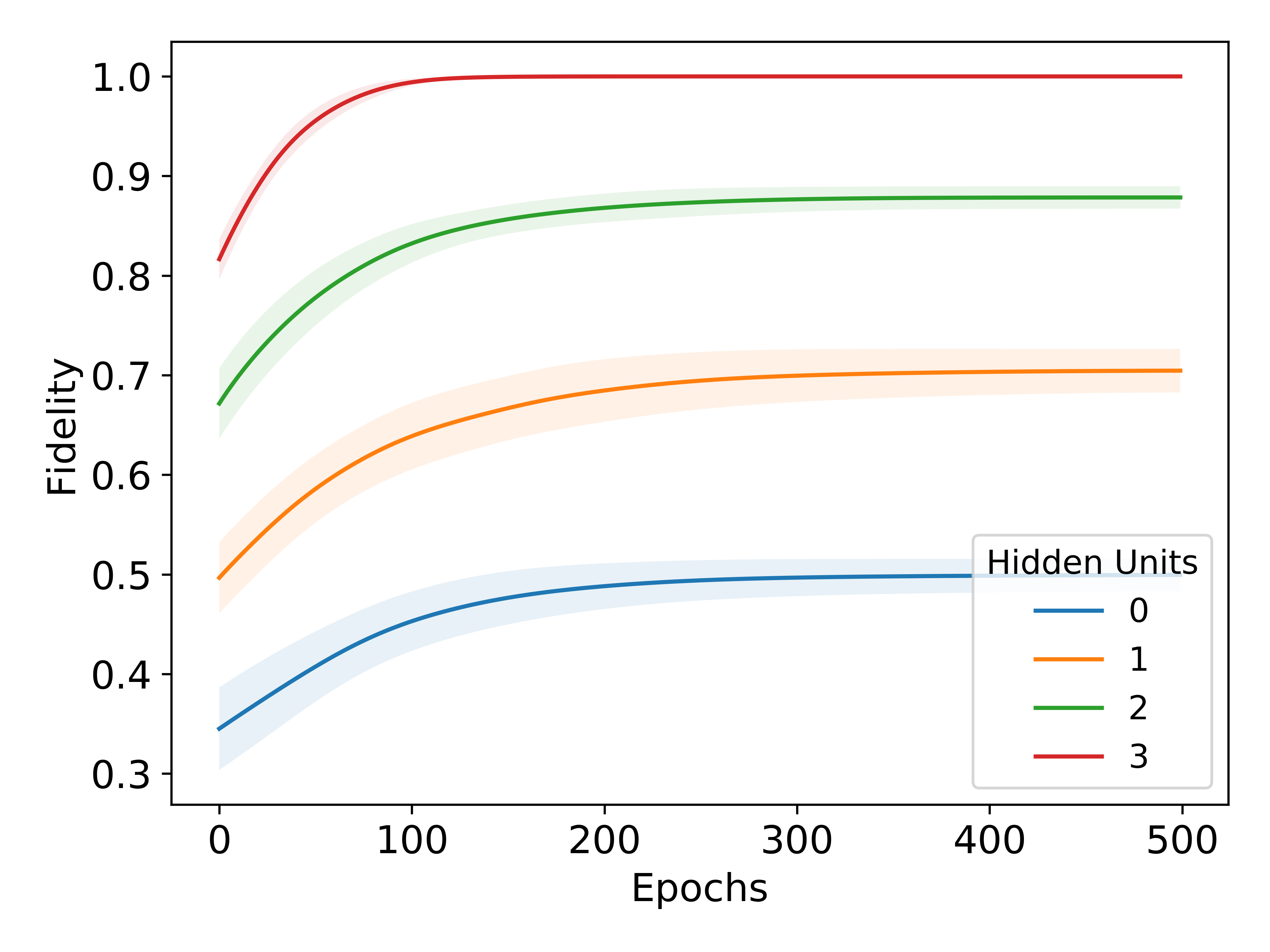}%
  }
 
 \caption{We trained the Unitary model with three visible units and an increasing number of hidden units. The target state is a random thermal state. The solid lines represent the average epoch value and the width of the shaded area two standard deviations over 50 runs. (a) Training loss (i.e. R\'enyi Divergence) of our model. (b) Fidelity between the target state and our model.}
  \label{UniNet_num}
\end{figure}

In addition to our theoretical results, we ran a series of small scale numerical experiments to showcase our ability to learn a thermal state with the Unitary model using the analytical gradients calculated in equation \eqref{gradient_thermal} and the ADAM optimizer with a learning rate of $0.001$ \cite{kingma2014adam}.  We constructed our target thermal states using a random two-local Hamiltonian model on $n_v$-qubits, i.e, 
\begin{equation}
   H_2 = \sum_i\sum_a J_a^{i} \sigma_a^i + \sum_{i<j} \sum_{a,b} J_{a, b}^{i, j} \sigma_a^i \sigma_b^j \label{two-local_ham}
\end{equation} 
where $\sigma_a^j=I^{\otimes j-1}\otimes \sigma_a \otimes I^{\otimes {n_v}-j}$, $\sigma_a$ are the Pauli matrices, and $J^i_a, J^{i,j}_{a,b} \in \mathbb{R}$ for $a, b\in \{x,y,z\}$.
For the experiments in Figure \ref{UniNet_num} and Supplementary Material \ref{app:additional_exp}, we constructed a Unitary model (see equation  \ref{state_prep}) by exhaustively sampling from the Hermitian terms in equation \eqref{two-local_ham}. We initialized our network coefficients by sampling from a normal distribution with mean $0$ and variance $1$, i.e. $\mathcal{N}(0,1)$. We initialize our target state by the sampling $J_{a}^{i}\sim \mathcal{N}(0,0.1)$, $J_{a, b}^{i, j}\sim \mathcal{N}(0,1)$ and then divided the resulting Hamiltonian by its operator norm. This initialization produced typical target states with a Volume Law scaling on the entanglement entropy, as discussed in \cite{marrero2020entanglement}, and prevented numerical instabilities when computing equation \eqref{gradient_thermal}.
In \cite{marrero2020entanglement} the authors observed gradient decay for these class of states as a function of the hidden units and our results show the absence of this phenomenon.     

Figure \ref{UniNet_num} demonstrates our ability to learn an ensemble of thermal states as we increase the number of hidden units of our Unitary model. For a Unitary model with three visible and three hidden units, we archived a $99\%$ average fidelity over our ensemble after $100$ epochs. \supp \ref{large_unitary_exp} contain a similar study when repeating this experiment on a larger system. Moreover, {\supp} \ref{hyperparamers_exp}, shows that we can successfully learn a fixed thermal state irrespective of the order of the gate operations we perform and network initialization. 

In all our experiments we saw no evidence of gradient decay or barren plateaus during training. Moreover, our model fidelity continued to increase with an increasing number of hidden units, suggesting an absence of entanglement induced Barren Plateaus.

\section{Application: Hamiltonian Learning}
\label{sec:ham}

As an example of how generative models like the QBM can be practical, we will apply them to the task of performing likelihood-free Hamiltonian learning~\cite{granade2012robust,wiebe2015quantum,kieferova2017tomography,anshu2021sample}, one of the few known significant quantum learning problems attaining a large quantum advantage.  The problem of quantum Hamiltonian learning is the problem of learning a generator for real-time  (or imaginary time) dynamics for a quantum system.  These operators allow us for example to control and calibrate quantum systems and many other tasks.  

\begin{figure}[t]
  \centering%
  \subfloat[Loss]{%
    \includegraphics[width=0.48\textwidth]{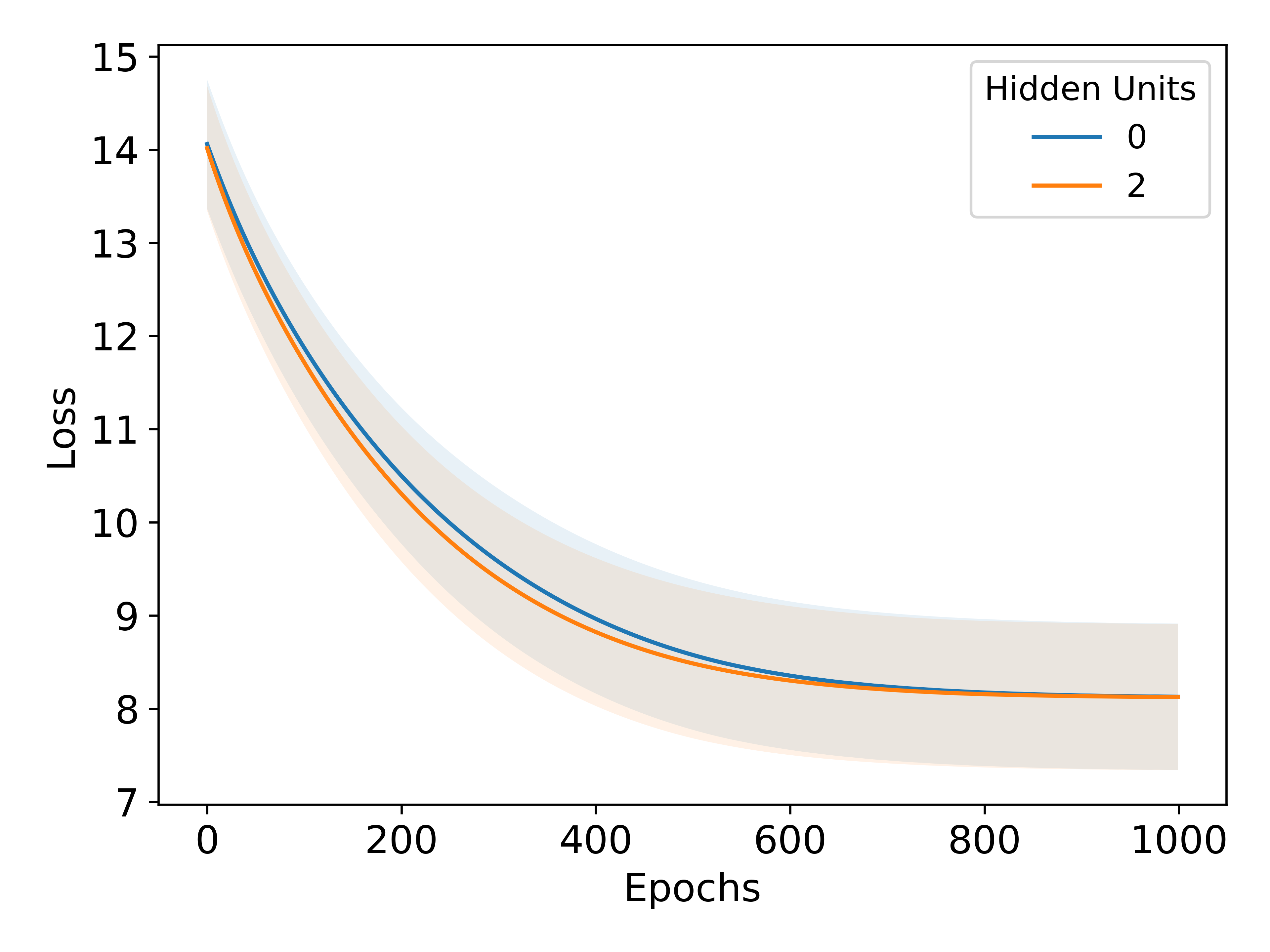}%
  }%
    \subfloat[Fidelity]{%
    \includegraphics[width=0.48\textwidth]{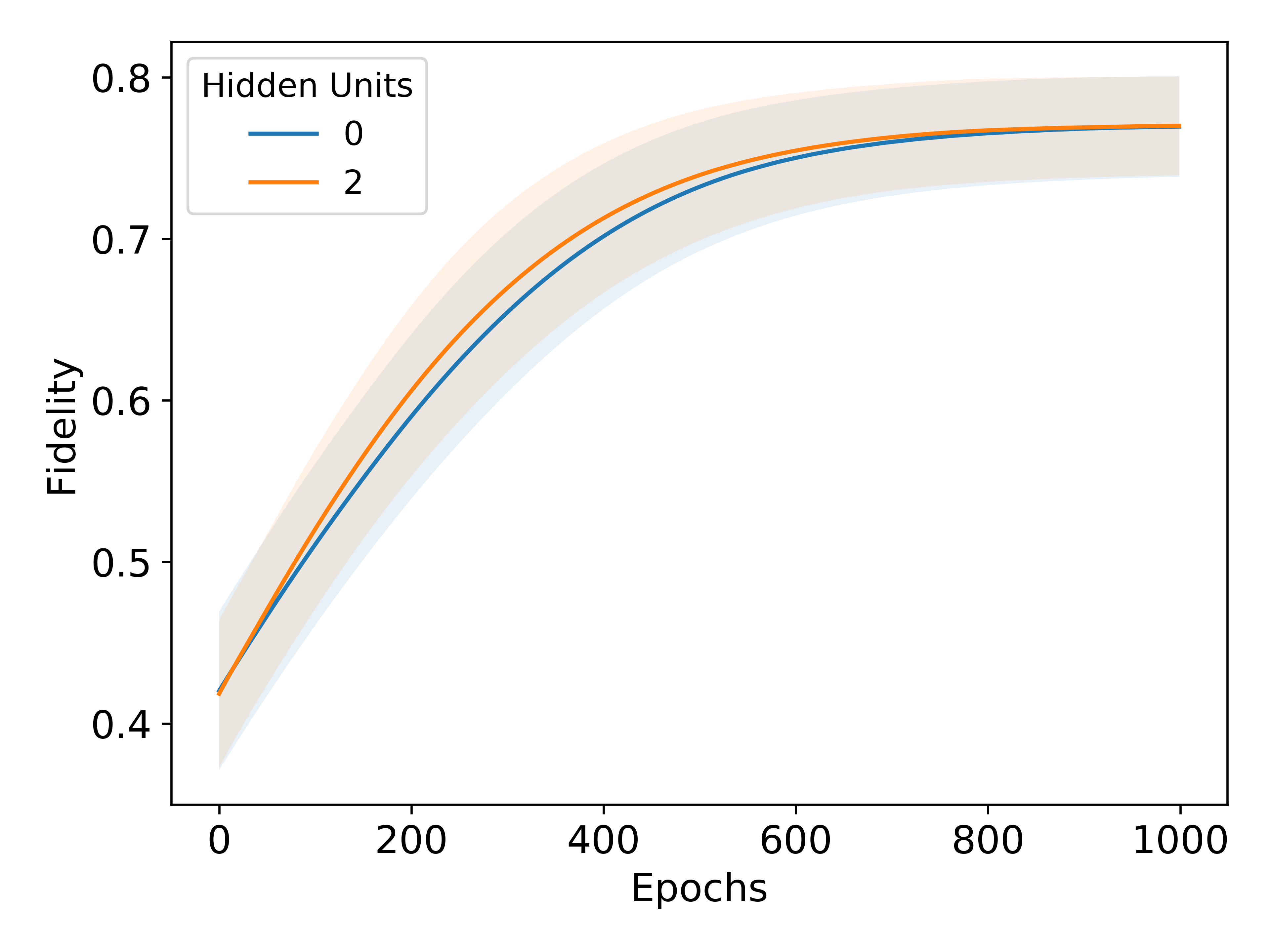}%
  }
 
 \caption{We trained the QBM with four visible units and an increasing number of hidden units. The target states are random thermal states using a three-local Hamiltonian with $\tau=10$. The solid lines represent the average epoch value and the width of the shaded area two standard deviations over an ensemble of 50 different target states. (a) Training loss (i.e. R\'enyi Divergence) and (b) Fidelity between the target states and our model.}
  \label{qbm_num}
\end{figure}

Here we investigate how effective are QBMs at solving problems in this space.  One advantage of divergence training is that it allows us to perform Hamiltonian learning without likelihood evaluations.

We devised a series of numerical experiments where the goal was to learn a target thermal state using a Hamiltonian from outside our model class. We optimized the model parameters using the analytical gradients calculated in \eqref{gradient_bm}. Our target state was constructed from a random three-local Hamiltonian on $n_v$-qubits, i.e.
   $$H_3 = H_2 + \sum_{i<j<k} \sum_{a,b,c} J_{a, b, c}^{i, j, k} \sigma_a^i \sigma_b^j \sigma_c^k$$
where $H_2$ is defined in \eqref{two-local_ham}, $\sigma_a^j=I^{\otimes j-1}\otimes \sigma_a \otimes I^{\otimes {n_v}-j}$, $\sigma_a$ are the Pauli matrices, and $J^i_a, J^{i,j}_{a,b}, J_{a, b, c}^{i, j, k}\in \mathbb{R}$ for $a,b,c\in \{x,y,z\}$. We initialized the Hamiltonian coefficients by sampling from $\mathcal{N}(0, 1)$ and then multiplying the resulting Hamiltonian by $\frac{\tau}{||H_3||}$, where $\tau$ is the temperature parameter, and $||H_3||$ is the operator norm of the Hamiltonian.  We used $H_2$ in \eqref{two-local_ham} on $(n_v+n_h)$-qubits, as our QBM Hamiltonian. We initialized the coefficients by sampling from $\mathcal{N}(0, 1)$ and then multiplied the resulting Hamiltonian by $\frac{1}{||H_2||}$. We optimized the model parameters using the ADAM optimizer with a learning rate of $0.001$ and imposed an L2-regularization with a penalty of coefficient of $2$ to prevent degenerative gradients. \supp~\ref{bm_l2_grad} contain more information about the effects of L2-regularization on the $\infty$-norm of our gradients.

Figure \ref{qbm_num} shows our ability to train a QBM for the ensemble of thermal states outlined above. Starting with a $42\%$ average fidelity, we managed to train QBMs that achieved a $77\%$ average fidelity after $1000$ epochs. \supp~\ref{bm_high_temp} contain a similar analysis on an ensemble of thermal states with a higher temperature i.e. smaller value of $\tau$. 

Our work represents fundamental research into the generic techniques of learning in the quantum realm. As such, it is too early to speculate about its impact on society. 

We did not observe a substantial benefit here from adding small numbers of hidden units to our model relative to the standard deviation in the performance of each of the instances studied.  In contrast, for the thermal state generation task using unitary networks, we noticed substantial advantages.  Much of these differences arise from the fact that the QBM yields mixed states rather than pure states and so there is less of a need in this context to use hidden units since even an all visible Boltzmann machine is capable of generating mixed states, unlike the unitary model.

\section{Conclusion}
Our work shows that choosing an unbounded loss function such as maximal R\'enyi divergence of order two allows one to avoid no-go theorems on barren plateaus and train unitary QNNs and QBMs. 
Compared to other unbounded loss functions such as relative entropy, gradients of this loss function have a closed-form and they can be in many cases efficiently estimated. By violating an assumption behind the barren plateau results, we were able to efficiently learn thermal states without experiencing gradient decay. We also demonstrated that maximal R\'enyi divergence can be used for Hamiltonian learning.
This approach to training opens new possibilities for training quantum model.

Our work focused strictly on generative training but could be applied in a broader context as pre-training step. In this approach, one would train a quantum model to perform a generative task first and then train it to perform a specific task in a second phase. Pre-training would ensure that the model at the start of a second phase is not random and thus does not suffer from a concentration of measure gradient decay.

The quantum circuits that arise in our work are not strictly shallow and suited for more advanced quantum hardware that might come on the market in the next years. Nevertheless, one could apply our sampling-based algorithm to learn simple quantum states using existing hardware which would lead to noisy implementation. Our training does not satisfy the assumptions leading to noise-induced barren plateaus~\cite{wang2020noise} but it does increase the circuit depth for training. Thus, the robustness of our training under noise is still an open question.

It should be also pointed out that our work does not guarantee an efficient training algorithm despite the fact that the gradient is efficiently computable and large at initialization. Specifically, convergence might be slow depending on the optimization landscape; however, the ease with which our algorithm learns and the absence of obvious barren plateaus in the landscape suggest that our approach is more practical and scalable than existing methods. The limitations of our training method are reflected in learning thermal states where our method becomes inefficient for low temperature states. This is to be expect because preparing ground states (corresponding to zero temperature) is QMA hard.  Nonetheless, quantum approaches like the D\" urr-H\o yer algorithm~\cite{durr1996quantum,gilyen2019quantum} can somewhat ameliorate these problems, a complete understanding of the cost of gradient descent for non-convex loss functions remains elusive.

We anticipate no immediate negative social impacts of the research or its application to Hamiltonian Learning or generation of thermal states owing to the low technological readiness of quantum computing at the moment although questions of fairness as well as privacy and security within quantum machine learning remain important avenues of inquiry.

This work leaves a number of open questions.  First, while we show that barren plateaus do not occur (under reasonable assumptions), our result is not able to prove a concentration of measure. Another question involves the role that these ideas may have in quantum generative adversarial networks.  Since our approach side steps the barren plateau problem, it may be ideal to deploy in such problems and a detailed study will be needed to understand the ramifications of this.  Finally, there is the question of whether these approaches can be used to generatively pretrain models for discriminative loss function (such as mean classification error rate), which suffer from the barren plateau problem because of their linear loss function.  This raises the intriguing possibility that generative training may have a role beyond quantum Hamiltonian learning: it may be an indispensable tool for avoiding gradient decay for large quantum neural networks.

\section{Acknowledgement}
MK was supported by the Sydney Quantum Academy, Sydney, NSW, Australia and  ARC Centre of Excellence for Quantum Computation and Communication Technology (CQC2T), project number CE170100012.
Support for C. Ortiz Marrero was provided by the Laboratory Directed Research and Development Program at Pacific Northwest National Laboratory, a multi-program national laboratory operated by Battelle for the U.S. Department of Energy, Release No. PNNL-SA-162897.
NW was funded by a grant from Google Quantum AI, and his theoretical work on this project was  supported by the U.S. Department of Energy, Office of Science, National Quantum Information Science Research Centers, Co-Design Center for Quantum Advantage under contract number DE-SC0012704
We thank Marco Tomamichel and Jarrod McClean for insightful comments.

\bibliographystyle{plainnat}  
\bibliography{ms}
%\newpage
%checklist
\blank{
 \begin{enumerate}
\item For all authors...
\begin{enumerate}
  \item Do the main claims made in the abstract and introduction accurately reflect the paper's contributions and scope?
    {\color{blue} Yes}
  \item Did you describe the limitations of your work?
    {\color{blue} Yes, in the conclusion and Supplementary materials}
  \item Did you discuss any potential negative societal impacts of your work?
   {\color{blue} Yes.}
  \item Have you read the ethics review guidelines and ensured that your paper conforms to them?
    {\color{blue} Yes}
\end{enumerate}

\item If you are including theoretical results...
\begin{enumerate}
  \item Did you state the full set of assumptions of all theoretical results?
    {\color{blue} Yes}
	\item Did you include complete proofs of all theoretical results?
    {\color{blue} Yes, in Supplementary materials}
\end{enumerate}

\item If you ran experiments...
\begin{enumerate}
  \item Did you include the code, data, and instructions needed to reproduce the main experimental results (either in the supplemental material or as a URL)?
    {\color{blue} Yes}
  \item Did you specify all the training details (e.g., data splits, hyperparameters, how they were chosen)?
    {\color{blue} Yes}
	\item Did you report error bars (e.g., with respect to the random seed after running experiments multiple times)?
    {\color{blue} Yes}
	\item Did you include the total amount of compute and the type of resources used (e.g., type of GPUs, internal cluster, or cloud provider)?
    {\color{blue} Yes, more details in the repository.}
\end{enumerate}

\item If you are using existing assets (e.g., code, data, models) or curating/releasing new assets...
\begin{enumerate}
  \item If your work uses existing assets, did you cite the creators?
    {\color{blue} Yes, in the paper and repository.}
  \item Did you mention the license of the assets?
    {\color{blue} Yes, BSD-2.}
  \item Did you include any new assets either in the supplemental material or as a URL?
    {\color{blue} Yes}
  \item Did you discuss whether and how consent was obtained from people whose data you're using/curating?
    {\color{blue} Yes}
  \item Did you discuss whether the data you are using/curating contains personally identifiable information or offensive content?
    \textcolor{blue}{N/A}
\end{enumerate}

\item If you used crowdsourcing or conducted research with human subjects...
\begin{enumerate}
  \item Did you include the full text of instructions given to participants and screenshots, if applicable?
    \textcolor{blue}{N/A}
  \item Did you describe any potential participant risks, with links to Institutional Review Board (IRB) approvals, if applicable?
   \textcolor{blue}{N/A}
  \item Did you include the estimated hourly wage paid to participants and the total amount spent on participant compensation?
  \textcolor{blue}{N/A}
\end{enumerate}
\end{enumerate}

%%%%%%%%%%%%%%%%%%%%%%%%%%%%%%%%%%%%%%%%%%%%%%%%%%%%%%%%%%%%

%\newpage
}
\appendix

\section{Proof of Theorem~\ref{thm:QNN}: Gradients for QNNs}\label{sec:proof_qnn}
Assuming we have a unitary QNN we can express the derivative of the density operator $\sigma$ as
\begin{align}
    \partial_{\theta_k} \sigma &={-i(\prod_{j=1}^{k-1} e^{-iH_j \theta_j} H_k \prod_{j=k}^L e^{-iH_j \theta_j} \ket{0}\!\bra{0}\prod_{j=L}^1 e^{iH_j \theta_j})+{\rm h.c.}}\nonumber\\
    &={-i(\prod_{j=1}^{k-1} e^{-iH_j \theta_j} H_k \prod_{j=k-1}^1 e^{iH_j \theta_j} \prod_{j=1}^L e^{-iH_j \theta_j} \ket{0}\!\bra{0}\prod_{j=L}^1 e^{iH_j \theta_j})+{\rm h.c.}}\nonumber\\
    &={-i(\prod_{j=1}^{k-1} e^{-iH_j \theta_j} H_k \prod_{j=k-1}^1 e^{iH_j \theta_j} \sigma)+{\rm h.c.}}=-i[\widetilde{H}_k, \sigma],\label{eq:commutator}
\end{align}
where $\widetilde{H}_k = \prod_{j=1}^{k-1} e^{-iH_j \theta_j} H_k \prod_{j=k-1}^1 e^{iH_j \theta_j}$. Therefore, because the derivative of the partial trace is the partial trace of the derivative
\begin{equation}
    \partial_{\theta_k} \widetilde{D}_2(\sigma_v\|\rho) = \frac{-i{\rm Tr}\left(\left\{{\rm Tr_h}({[\widetilde{H}_k, \sigma]}),\sigma_v\right\} \rho^{-1}\right)}{{\rm Tr}\left(\sigma_v^2 \rho^{-1}\right)}. \label{gradient_trotter}
\end{equation}

We can also compute the forward divergence $\partial_{\theta_k} \widetilde{D}_2(\rho\|\sigma_v)$ 

\begin{equation}
    \partial_{\theta} \widetilde{D}_2(\rho\|\sigma_v(\theta)) = \frac{- {\rm Tr}\left(\rho^2 \sigma_v^{-1} (\partial_\theta \sigma_v(\theta))\sigma^{-1}\right)}{{\rm Tr}\left(\rho^2 \sigma_v^{-1}\right)} 
    =  \frac{i {\rm Tr}\left(\rho^2 \sigma_v^{-1} {\rm Tr}_h(
    [\widetilde{H}_k, \sigma]
    )\sigma_v^{-1}\right)}{{\rm Tr}\left(\rho^2 \sigma_v^{-1}\right)}. \label{eq:forwardDiff}
\end{equation}
These results together prove the theorem.

The clear advantage of this approach is that both forward and reverse R\'enyi divergence have a closed form.  In Section~\ref{sec:thermal} we show how we can use~\eqref{gradient_uni_rho} for thermal state generation. 
Note that if one is only interested in computing the direction of the gradient and not the magnitude, it is sufficient to only compute the nominators in~\eqref{gradient_uni_sigma} and~\eqref{gradient_uni_rho}. 

This leads to the following theorem:
\unitarygradient*

\section{Extended swap test}\label{sec:proof_swap}

Here we discuss the main algorithmic component of training. We combine the extended version of the swap test with Hadamard test to estimate the quantity $Re[{\rm Tr} (\prod_{i=1}^n \prod_{j=1}^m U_j \rho_i)]$ for unitaries $U_i$ and quantum states $\rho_i$.

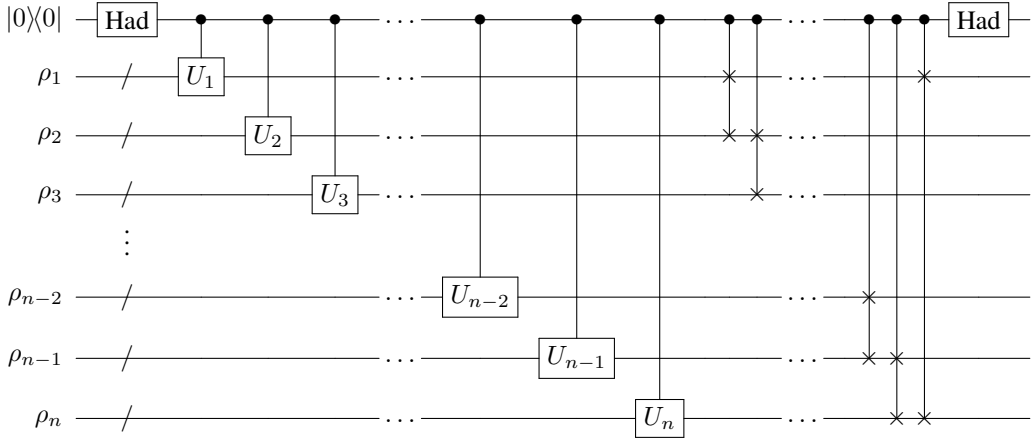
\begin{figure}[t]
    \centering
    $$\Qcircuit @C=0.8em @R=0.8em {
   \lstick{\ket{0}\!\!\bra{0}}&\gate{\text{Had}}&\ctrl{1}&\ctrl{2}& \ctrl{3} &\qw&\dots&  & \ctrl{6}& \ctrl{7}&\ctrl{8}&\qw&\ctrl{2}&\ctrl{3}&\qw&\dots&  &\qw& \ctrl{7} &\ctrl{8} &\ctrl{8}&\gate{\text{Had}}&\qw\\
\lstick{\rho_1}&{/}\qw&\gate{U_1}&\qw& \qw  &\qw&\dots&& \qw& \qw&\qw&\qw&\qswap  &\qw   &\qw&\dots&   &\qw& \qw& \qw &\qswap&\qw&\qw \\
\lstick{\rho_2}&{/}\qw&\qw&\gate{U_2}& \qw  &\qw&\dots& & \qw&\qw&\qw&\qw&\qswap  &\qswap   &\qw&\dots&  &\qw& \qw& \qw & \qw& \qw&\qw  \\
\lstick{\rho_3}&{/}\qw&\qw&\qw& \gate{U_3}  &\qw&\dots& & \qw& \qw&\qw&\qw&\qw    &\qswap &\qw&\dots&    &\qw  &\qw&\qw& \qw & \qw& \qw \\
&\vdots  &&     &&   &   & &&&  & && \\
&&&&&&&&\\
\lstick{\rho_{n-2}}&{/}\qw&\qw&\qw& \qw  &\qw&\dots&& \gate{U_{n-2}}&\qw& \qw&\qw& \qw    & \qw     &\qw&\dots&  &\qw& \qswap & \qw & \qw& \qw&\qw  \\
\lstick{\rho_{n-1}}&{/}\qw&\qw&\qw& \qw &\qw&\dots&&\qw& \gate{U_{n-1}}&\qw&\qw& \qw    & \qw &\qw &\dots&   &\qw& \qswap& \qswap & \qw& \qw&\qw \\
\lstick{\rho_n}&{/}\qw& \qw&\qw&\qw& \qw&\dots& &\qw & \qw& \gate{U_n}&\qw&\qw    & \qw     &\qw& \dots&& \qw& \qw& \qswap& \qswap& \qw &\qw
}$$
    \caption{This unitary, denoted $\chi$, consists of two Hadamard gates, a series and of controlled unitaries and the controlled cyclic permutation $C_+$.  If we wish to uncontrollably permute $p$ registers each consisting of $q$ qubits, the circuit would require $pq$ controlled swap operations which can be decomposed into $pq$ Toffoli gates and $2pq$ CNOT gates. At the end of the test, the first qubit is measured.}
    \label{fig:controlled_swap}
\end{figure}

\swapamard*

The proof follows from a generalization of the swap-test and Hadamard test.
First, $S_+$ be a cyclic permutation operator such that for any bounded operator of the form $A_1,\ldots, A_n$, 
\begin{equation}
    S_+ A_1 \otimes \cdots \otimes A_n S_+^\dagger = A_n\otimes A_1 \otimes \cdots \otimes A_2\label{eq:Sdef}
\end{equation}

Next, let $C_+$ be a controlled cyclic permutation of the form
\begin{equation}
    C_+ := \ket{0}\!\bra{0} \otimes I + \ket{1}\!\bra{1} \otimes S_+
\end{equation}
The operation $C_+$ can be implemented with a series of nested, controlled swaps.

It follows from algebra that
\begin{align}
    & \left(\ket{0}\!\bra{0} \otimes I\right)\chi \left( \ket{0}\!\bra{0} \otimes \rho_1 \otimes \dots \otimes \rho_n \right) \chi^{\dagger}  \nonumber\\
    &\qquad= \frac{1}{4} \ket{0}\!\bra{0} \otimes \left(\bigotimes_{i=1}^n \rho_i + S_+\bigotimes_{i=1}^n U_i \rho_i + \bigotimes_{i=1}^n \rho_i U_i^{\dagger} S_+^\dagger + S_+ \bigotimes_{i=1}^n U_i \rho_i U_i^{\dagger} S_+^\dagger\right).
\end{align}
It then follows from~\eqref{eq:Sdef} and the fact that a density matrix is unit-normalized that
\begin{equation}
    {\rm Tr}\left(\left(\ket{0}\!\bra{0} \otimes I\right)\chi \right) = \frac{1}{4}\left( 2 + {\rm Tr}\left(S_+\bigotimes_{i=1}^n U_i \rho_i + \bigotimes_{i=1}^n \rho_i U_i^{\dagger} S_+^\dagger \right) \right).\label{eq:swapeqn}
\end{equation}
Next note that by the definition of the cyclic shift we have that
\begin{equation}
    \bra{x_1,\ldots, x_n} S_+ \ket{y_1,\ldots, y_n} = \delta_{x_1,y_n}\delta_{x_2,y_1}\cdots \delta_{x_n,y_{n-1}}.
\end{equation}
We then have that
\begin{align}
    {\rm Tr}\left(S_+\bigotimes_{i=1}^n U_i \rho_i \right) &= \sum_{i_1,\ldots, i_n} \bra{i_1,\ldots,i_n} S_+ \bigotimes_{i=1}^n U_i\rho_i \ket{i_1,\ldots, i_n}\nonumber\\
    &=\sum_{i_1,\ldots, i_n} \sum_{j_1,\ldots, j_n} \bra{i_1,\ldots,i_n} S_+ \ket{j_1,\ldots,j_n}\!\bra{j_1,\ldots,j_n} \bigotimes_{i=1}^n U_i\rho_i \ket{i_1,\ldots, i_n}\nonumber\\
    &=\sum_{i_1,\ldots, i_n} \bra{i_1,\ldots,i_n}\sum_{j_1,\ldots, j_n} \delta_{i_1,j_n}\delta_{i_2,j_1}\cdots \delta_{i_n,j_{n-1}}\bra{j_1,\ldots,j_n} \bigotimes_{i=1}^n U_i\rho_i \ket{i_1,\ldots, i_n}\nonumber\\
    &=\sum_{i_1,\ldots, i_n} \bra{i_n}U_i\rho_1 \ket{i_1}\bra{i_1}U_2\rho_2 \ket{i_2}\cdots \bra{i_{n-1}}U_n\rho_n \ket{i_n} \nonumber\\
    &=\sum_{i_n} \bra{i_n}U_1\rho_1 U_2\rho_2 \cdots U_n\rho_n \ket{i_n}={\rm Tr}\left(\prod_{i=1}^n U_i\rho_i \right).
\end{align}
By repeating the exact same calculation, we also find that
\begin{equation}
    {\rm Tr}\left(\bigotimes_{i=1}^n \rho_i U_i^{\dagger} S_+^\dagger\right)= {\rm Tr}\left(\prod_{i=1}^n \rho_i U_i^{\dagger} \right).
\end{equation}
We therefore have that
\begin{equation}
    \frac{1}{4}\left( 2 + {\rm Tr}\left(S_+\bigotimes_{i=1}^n U_i \rho_i + \bigotimes_{i=1}^n \rho_i U_i^{\dagger}S_+^\dagger \right) \right) =\frac{1+ \text{Re}[{\rm Tr} (\prod_{i=1}^n U_i\rho_i)]}{2},
\end{equation}
and therefore the result follows from~\eqref{eq:swapeqn}.

\begin{corollary}[Density matrix power]
Let $O_\rho$ be an oracle that prepares $\ket{\psi}$ in a Hilbert space $\mathcal{H}_A\otimes \mathcal{H}_B$ such that ${\rm Tr}_A (\ket{\psi}\!\bra{\psi}) = \rho \in \mathbb{C}^{2^n\times 2^n}$.  
\begin{enumerate}
    \item If $O_{\rho}:\sigma \mapsto \sigma \otimes \ket{\psi} \!\bra{\psi}$ then for any $m$, there exists an algorithm that can estimate ${\rm Tr}(\rho^m)$ within error $\epsilon$ with probability greater than $2/3$ using $O(m/\epsilon^2)$ queries to $O_{\rho}$. Furthermore, $\Omega(m)$ queries to $O_{\rho}$ are needed for any blackbox algorithm that can estimate ${\rm Tr}(\rho^m)$ within error less than $1/3$.
    \item If $O_{\rho}:\ket{0} \mapsto \ket{\psi} $ is a unitary operation with a known inverse then for any $m$, there exists an algorithm that can estimate ${\rm Tr}(\rho^m)$ within error $\epsilon$ with probability greater than $2/3$ using $O(m/\epsilon)$ queries to $O_{\rho}$. Furthermore, $\Omega(\sqrt{m})$ queries to $O_{\rho}$ are needed for any blackbox algorithm that can estimate ${\rm Tr}(\rho^m)$ within error less than $1/3$.
\end{enumerate}
\end{corollary}\label{col:matrix_powers}

\begin{proof}
Case 1) The upper bound follows directly from statistical sampling and the circuit of~\ref{thm:swap_trick} and standard bounds on the variance of the mean and then applying Chebyshev's inequality.  The lower bound follows from state discrimination bounds.

Density matrix power protocol is well known in the literature. In situations when qubits can be reused after measurement, \cite{yirka2020qubit} gives a protocol where the number of qubits does not depend on $n$.

First, let us define an isometry superoperator $I$ such that $I: U \mapsto U\otimes \mathds{1}$ for any matrix $U$ and $\mathds{1}\in \mathbb{C}^{2^n\times 2^n}$.  It is then clear that for any operator $U$ and density operator $\sigma$ in $\mathbb{C}^{2^n\times 2^n}$, $I(U) O_{\rho}(\sigma) I(U^\dagger) = O_{\rho}(U\sigma U^\dagger)$.  Therefore without loss of generality, we can always assume that state is initialized to the state $\rho^{\otimes n} = O_{\rho}^{n}(1)$ which requires $n$ queries.  Put in other words, we can always assume that all the queries are performed at the beginning of the protocol given our choice of oracle.

This fact makes state discrimination bounds a natural way to understand the limitations of this protocol since the number of copies used in the optimal measurement to discriminate the states is equivalent to the number of queries for this oracle.  Let us take $\ket{\phi}$ to be an arbitrary pure state of dimension $N\ge 2$.  We then consider the task of discriminating between the states
\begin{align}
    \sigma_0 &= \ket{\phi} \!\bra{\phi}\nonumber\\
    \sigma_1 &= (1-\epsilon)\ket{\phi} \!\bra{\phi} + \frac{\epsilon\mathds{1}}{N},
\end{align}
for some $\epsilon >0$.  We can discriminate between the two states by noting that ${\rm Tr}(\sigma_0^n) = 1$ for all $n$ but we see from the binomial theorem that 
\begin{equation}
    {\rm Tr}(\sigma_1^n) = (1- \epsilon + \epsilon/N)^n + (N-1) \left(\frac{\epsilon}{N}\right)^n = 1-n\left(1-\frac{1}{N}\right)\epsilon + O(\epsilon^2)\label{eq:lb1}
\end{equation}
Thus if we could compute ${\rm Tr}(\sigma^n)$ within error $1/3$ and success probability $2/3$ then we could distinguish the two states given that $n$ is chosen to be in $\Theta(1/\epsilon)$. 

The optimal failure probability of any POVM discriminating between these states given an $n$-copy state is given by the quantum relative entropy between $\sigma_0$ and $\sigma_1$~\cite{vedral2002role}:
\begin{equation}
    P_{\rm fail} = e^{-n S(\sigma_0||\sigma_1)}\label{eq:lb0}
\end{equation}
In this case the relative entropy between the two states is
\begin{align}
    S(\sigma_0||\sigma_1) &= {\rm Tr}(\sigma_0 \ln \sigma_0) - {\rm Tr}(\sigma_0 \ln \sigma_1) \nonumber\\
    &= - {\rm Tr}(\sigma_0 \ln \sigma_1) = - \ln((1-\epsilon) + \epsilon/N) \le \ln\left(\frac{1}{1-\epsilon}\right)= \epsilon + O(\epsilon^2).\label{eq:lb2}
\end{align}
Thus in order for the success probability, optimally to be $1/3$, it is necessary to pick $n \in \Omega(1/\epsilon)$.  We then see that if the protocol given in~\eqref{eq:lb1} required $o(n) = o(1/\epsilon)$ queries to prepare the state then we would violate the bound in~\eqref{eq:lb0} based on the relative entropy calculation in~\eqref{eq:lb2}.  Therefore the query complexity of computing ${\rm Tr}(\rho^n)$ within constant error using these non-unitary oracles is in $\Omega(n)$

Case 2) The proof of the upper bound is straightforward.  From Theorem~\ref{thm:swap_trick}, a circuit exists that uses $m$ queries to $O_{\rho}$ and such that the marked qubit is $0$ with probability $\frac{1+ Re[{\rm Tr} (\rho^m)]}{2}$.  Since $O_{\rho}$ is unitary and invertible, amplitude estimation can be used to learn the value~\cite{brassard2002quantum} within error $\epsilon$ and probability of failure at most $1/3$ using $O(1/\epsilon)$ applications of the circuit in~\ref{thm:swap_trick}.  Since this circuit uses $O(m)$ queries, there exists a quantum algorithm that requires $O(m/\epsilon)$ queries to $O_\rho$ that outputs the trace within the desired uncertainties.  

The lower bound follows by contradiction.  Let us consider the search problem.  We can solve the search by searching half-spaces for the presence of a marked item.  This search can be done by using a single iteration of Grover's algorithm and then converting the result into a mixed state by sampling and forgetting the result. 
In particular, let the initial state for Grover's search be $\ket{\psi}$ and assume that the marked state is $\ket{m}$.  At one iteration of the Grover oracle, the state is of the form
\begin{equation}
    \ket{\psi'} := \sin(3\theta) \ket{m} + \cos(3\theta) \ket{\psi^\perp}.
\end{equation}
where $\theta = \arcsin(\sqrt{\frac{M}{N}})$ for the case of $N$ items and $M$ marked items in the set.

If there are $L$ items that we are searching over.  We reduce the problem to a trace estimation problem by dividing the domain in half and determining whether $m\in [0,L/2)$ or its complement.  We need to consider two cases, either $M=0$, which corresponds to no marked entry in the half space or $M=1$ which implies that there is a marked element in that space.  Let us first consider the case where $M=1$.  In this case $\theta=\arcsin(\sqrt{2/L})$ and we will convert it into a mixed state $\rho$ by performing a query to the Grover oracle $\mathcal{O}_g\ket{x} \ket{0} = \ket{x} \ket{\delta_{x,m}}$ with output on an ancillary qubit.  Then if we perform a partial trace over this qubit the state that we are left with is
\begin{equation}
    \sigma'_0 = \sin^2(3\theta) \ket{m}\!\bra{m} + \cos^2(3\theta) \ket{\psi^\perp}\!\bra{\psi^\perp}.
\end{equation}
On the other hand, if there is no marked state within the range we would see a density operator of the form
\begin{equation}
    \sigma_1' = \ket{\psi^\perp}\!\bra{\psi^\perp}.
\end{equation}
We can distinguish these two cases because, for and $n>0$, ${\rm Tr}({\sigma'}_0^n) < 1$ but ${\rm Tr}({\sigma'}_1^n) = 1$.  This allows us to reduce the problem of search to one of trace estimation.

In order to find the number of queries needed, we first need to estimate the power $n$ that would be required to lower the trace to at most a value of $2/3$. Then for any $n\ge 2$
\begin{align}
    {\rm Tr}(\sigma_0'^n)= \sin^{2n} (3\theta) + \cos^{2n}(3\theta) = \left(1-\frac{(3\theta)^2}{2!} \right)^{2n} + O(\theta^4)= 1- 9n\theta^2 +O(\theta^4).
\end{align}
Since $\theta \in \Theta(1/\sqrt{L})$ we therefore have that the trace can be lowered to $2/3$ by taking a value of $n\in \Theta(L)$.  

If we repeat this $O(\log(L))$ times.  Now let us assume that estimating the trace were possible within constant accuracy using $O(\sqrt{L})$ queries then the total number of queries would be in
\begin{equation}
    O(\sum_j \sqrt{L/2^j}) = O(\sqrt{L}).
\end{equation}
Thus if we follow this procedure starting at $L=N$ we could therefore solve the search problem in $O(\sqrt{N})$ queries, which is impossible~\cite{brassard2002quantum}.  Therefore the trace estimation process must require $\Omega(\sqrt{N})$ queries when provided the density operator in a purified form yielded by a unitary oracle.
\end{proof}

The use of the swap test for computing powers of density matrices is been established in the literature. Notably, ref.~\cite{yirka2020qubit} utilizes qubit resets to improve the spatial complexity of the algorithms from $O(nm)$ to $O(n)$.

The above results show that the extended swap test is asymptotically optimal if the input is provided by a non-unitary oracle.  However, if a purification is provided then the above result shows that there is a potential quadratic gap between the cost required by applying amplitude estimation on the result of the extended swap test and the lower bound.  Determining whether the lower bound is tight under these circumstances remains an open question.

%%%%%%%%%%%%%%%%%%%%%%%%%%%%%%%%%%%%%%%

\section{Efficient algorithm for computing the reverse divergence for the Unitary Network}\label{sec:proof_gradient_alg}

Our algorithm for estimating the gradient involves using the extended swap test in concert with amplitude estimation to estimate the numerator and denominator in the expression for the gradient.  We state the amplitude estimation result below for completeness.

\begin{theorem}[Amplitude Estimation~\cite{brassard2002quantum}]\label{thm:amp_est} \ \\
For any positive integers $k$, $M$ and unitary matrix $U$ and projector $P_{\rm mark}$ such that ${\rm Tr}(P_{\rm mark} U\ket{0}\!\bra{0}U^\dagger) =p$,
there exists an algorithm that
outputs $\tilde{p}: (0 \leq \tilde{p} \leq 1)$ such that
\begin{equation}
     \left| \tilde{p} - p \right| \;\leq\;
2 \pi k \frac{\sqrt{p(1-p)}}{M} +  k^2 \frac{\pi^2}{M^2} \label{ae}
\end{equation}
with
probability at least $\frac{8}{\pi^2}$ when $k=1$ and
with probability greater than $1-\frac{1}{2(k-1)}$ for $k
\geq 2$.  It uses $O(M)$ evaluations of $U$ and $I - 2 P_{\rm mark}$.
\end{theorem}
With this result in place, we can then prove our main for the gradient of the reverse divergence for the unitary network, which we restate below for convenience.
\generic*

\begin{proof}
We compute the gradient by estimating the numerator and then the denominator and return our estimate of the gradient to be the quotient of the two. Under the assumption that the error in the numerator is at most $|\delta_1|$ and the error in the denominator is at most $|\delta_2|$ where both quantities are assumed to have a relative error at most $1/2$, we have from the triangle inequality

\begin{align}
    &\left|\frac{{\rm Tr}\left(\left\{{\rm Tr_h}({[\widetilde{H}_k, \sigma]}),\sigma_v\right\} \rho^{-1}\right)}{{\rm Tr}\left(\sigma_v^2 \rho^{-1}\right)}-\frac{{\rm Tr}\left(\left\{{\rm Tr_h}({[\widetilde{H}_k, \sigma]}),\sigma_v\right\} \rho^{-1}\right)+\delta_1}{{\rm Tr}\left(\sigma_v^2 \rho^{-1}\right)+\delta_2}\right|\nonumber\\
    &\le \frac{|\delta_1|}{{\rm Tr}\left(\sigma_v^2 \rho^{-1}\right)}+ \left| {\rm Tr}\left(\left\{{\rm Tr_h}({[\widetilde{H}_k, \sigma]}),\sigma_v\right\} \rho^{-1}\right)+|\delta_1| \right|\left|\frac{1}{{\rm Tr}\left(\sigma_v^2 \rho^{-1}\right)} - \frac{1}{{\rm Tr}\left(\sigma_v^2 \rho^{-1}\right)+\delta_2} \right|\nonumber\\
    &\le \frac{|\delta_1|}{{\rm Tr}\left(\sigma_v^2 \rho^{-1}\right)}+ \frac{3}{2}\left| {\rm Tr}\left(\left\{{\rm Tr_h}({[\widetilde{H}_k, \sigma]}),\sigma_v\right\} \rho^{-1}\right) \right|\left|\frac{\delta_2}{{\rm Tr}^2\left(\sigma_v^2 \rho^{-1}\right) +\delta_2 {\rm Tr}\left(\sigma_v^2 \rho^{-1}\right) }  \right|\nonumber\\
    &\le \frac{|\delta_1|}{{\rm Tr}\left(\sigma_v^2 \rho^{-1}\right)}+ 3|\delta_2|\left| \frac{{\rm Tr}\left(\left\{{\rm Tr_h}({[\widetilde{H}_k, \sigma]}),\sigma_v\right\} \rho^{-1}\right)}{{\rm Tr}^2\left(\sigma_v^2 \rho^{-1}\right)} \right|\nonumber\\
    &= \frac{|\delta_1|}{{\rm Tr}\left(\sigma_v^2 \rho^{-1}\right)}+ 3|\delta_2|\left| \frac{{\rm Tr}\left(({[\widetilde{H}_k, \sigma]})\left\{\sigma_v, \rho^{-1}\right\} \otimes I/2^{n_h}\right)}{{\rm Tr}^2\left(\sigma_v^2 \rho^{-1}\right)} \right|\nonumber\\
    &\le \frac{|\delta_1|}{{\rm Tr}\left(\sigma_v^2 \rho^{-1}\right)}+ 3|\delta_2|\left| \frac{{\rm Tr}\left(\left\{\sigma,\left\{\sigma_v, \rho^{-1}\right\} \otimes I/2^{n_h}\right)\right\}}{{\rm Tr}^2\left(\sigma_v^2 \rho^{-1}\right)} \right|\le \frac{|\delta_1| + 12 |\delta_2|}{{\rm Tr}\left(\sigma_v^2 \rho^{-1}\right)}
 \label{tighter_bound}
\end{align}
Here we used that $\|\widetilde{H}_k\| = 1$ and $|\delta_2| \le 1/2$. Therefore to ensure an error of $\epsilon$ it suffices to take $|\delta_1| = |\delta_2|$ and
\begin{equation}
    |\delta_2| \le  \frac{\epsilon {\rm Tr}(\sigma_v^2 \rho^{-1})}{13} \label{eq:bds}
\end{equation}
where the sufficiency condition on the value of $\epsilon$ is therefore guaranteed by $\epsilon \le \frac{13}{2{\rm Tr}(\sigma_v^2 \rho^{-1})}.$

Having the relationships between the errors in the numerator, denominator and the entire gradient, we will now focus on estimating the gradient. 

Then we can use the extended swap test to sample from ${\rm Tr}[\sigma^2_v \rho^{-1}]$. A single run of the test  requires $O(N)$ queries to oracles $e^{i H_k \theta_k}$ to prepare $\sigma$, a single query to prepare $\rho^{-1}/{\rm Tr}(\rho^{-1})$ and $O(n_v)$ gates to preform swap. Estimating the denominator up to $\delta_2$ we need $O\left( \frac{N{\rm Tr}^2(\rho^{-1})}{\delta_2^2}\right)$ queries and  $O\left( \frac{n_v{\rm Tr}^2(\rho^{-1})}{\delta_2^2}\right)$ gates if we used sampling, but Theorem~\ref{thm:amp_est} allows us to reduce the complexity quadratically in terms of the error.

For our case, it is sufficient to take $k=1$ and consider the worst-case scenario $p=1/2$. Thus, we only need $O\left( \frac{N}{\delta_2}\right)$ queries and  $O\left( \frac{n_v}{\delta_2} \right)$ gates for the denominator. Similarly, we can estimate the numerator by expanding the numerator and the denominator into a linear combination of $4$ terms and estimating each one using the extended swap test and amplitude estimation. The term $\widetilde{H}_K$ requires $O(N)$ queries and thus the complexity of the numerator is also $O\left( \frac{N{\rm Tr}(\rho^{-1})}{\delta_1}\right)$ queries and  $O\left( \frac{n_v{\rm Tr}(\rho^{-1})}{\delta_1} \right)$ gates. This allows us to state the complexity of gradient computation as
\begin{equation}
    N_{\rm ops} = N_{\rm queries} + N_{\rm gates} \in O\left(\frac{n_v+N}{ \min\{|\delta_1|,|\delta_2|\}} \right) = O\left(\frac{n_v+N}{\epsilon \rm{Tr}(\sigma_v^2 \rho^{-1} / {\rm Tr}(\rho^{-1}))} \right).
\end{equation}

\end{proof}

\section{Training QBMs}
Here we provide a specific quantum algorithm for computing the derivative of the reverse R\'eyni divergence for Boltzmann machines.  First we provide a proof of the gradients for the reverse gradient and then provide a method that uses the extended swap test to evaluate the gradient in the subsequent subsection.
\subsection{Proof of Theorem~\ref{thm:boltz}: Gradients for QBMs}\label{sec:proof_qbm}
The aim of training a quantum Boltzmann machine here is then to minimize the R\'enyi Divergence $\widetilde{D}_2(\rho || \sigma(\theta))$ for training distribution $\rho$.  We further assume that the user has access to an oracle that is capable of preparing $\sigma(\theta)$ for any $\theta$ and wishes to use this oracle through a gradient descent process to minimize the R\'enyi Divergence. We show below a method that can be used to approximate the gradient through a commutator expansion.

\begin{align}
    \partial_\theta\sigma(\theta) &= \frac{{\rm Tr_h}\left(\partial_{\theta} e^{-H(\theta)} \right)}{{\rm Tr}\left( e^{-H(\theta)}\right)}-\frac{{\rm Tr_h}\left(e^{-H(\theta)} \right){\rm Tr}\left(\partial_{\theta} e^{-H(\theta)} \right)}{{\rm Tr}\left( e^{-H(\theta)}\right)^2} \nonumber\\
    &= \frac{{\rm Tr_h}\left(\partial_{\theta} e^{-H(\theta)} \right)}{{\rm Tr}\left( e^{-H(\theta)}\right)}-\frac{\sigma(\theta){\rm Tr}\left(\partial_{\theta} e^{-H(\theta)} \right)}{{\rm Tr}\left( e^{-H(\theta)}\right)} 
\end{align}

We therefore have that
\begin{align}
    &\partial_{\theta} D_2(\rho\|\sigma(\theta)) = \frac{- {\rm Tr}\left(\rho^2 \sigma^{-1}(\theta) (\frac{{\rm Tr_h}\left(\partial_{\theta} e^{-H(\theta)} \right)}{{\rm Tr}\left( e^{-H(\theta)}\right)})\sigma^{-1}(\theta)\right)}{{\rm Tr}\left(\rho^2 \sigma^{-1}\right)}+\frac{{\rm Tr}\left(\partial_{\theta} e^{-H(\theta)} \right)}{{\rm Tr}\left( e^{-H(\theta)}\right)}\nonumber\\
    &=\frac{- {\rm Tr}\left(\rho^2 \sigma^{-1}(\theta) (\frac{{\rm Tr_h}\left(\partial_{\theta} e^{-H(\theta)} \right)}{{\rm Tr}\left( e^{-H(\theta)}\right)})\sigma^{-1}(\theta)\right)}{{\rm Tr}\left(\rho^2 \sigma^{-1}\right)}-\frac{{\rm Tr}\left(\left(\partial_{\theta} H(\theta)\right)e^{-H(\theta)} \right)}{{\rm Tr}\left( e^{-H(\theta)}\right)}\nonumber\\
     &=\frac{ {\rm Tr}\left(\rho^2 \sigma^{-1}(\theta) ({{\rm Tr_h}\left(\int_0^1 e^{-sH(\theta)} (\partial_{\theta} H(\theta)) e^{-(1-s)H(\theta)}\right)})\sigma^{-1}(\theta)\right)}{{\rm Tr}\left(\rho^2 \sigma^{-1}\right){\rm Tr}\left( e^{-H(\theta)}\right)}-\frac{{\rm Tr}\left(\left(\partial_{\theta} H(\theta)\right)e^{-H(\theta)} \right)}{{\rm Tr}\left( e^{-H(\theta)}\right)}\nonumber\\
     &=\frac{ {\rm Tr}\left(\rho^2 \sigma^{-1}(\theta) ({{\rm Tr_h}\left(\int_0^1 e^{-sH(\theta)} (\partial_{\theta} H(\theta)) e^{-(1-s)H(\theta)}\right)})\sigma^{-1}(\theta)\right)}{{\rm Tr}\left(\rho^2 \sigma^{-1}\right){\rm Tr}\left( e^{-H(\theta)}\right)}-\frac{{\rm Tr}\left(\left(\partial_{\theta} H(\theta)\right)e^{-H(\theta)} \right)}{{\rm Tr}\left( e^{-H(\theta)}\right)}\nonumber\\
\end{align}
This expression can be simplified using Hadamard's lemma and reduce to a statistical sampling problem. The expression is valid, but requires preparations of the inverse of the partial trace which is not ideal.

If we assert that $\rho$ is maximum rank, we  can instead evaluate the reverse divergence
\begin{equation}
    D_2(\sigma(\theta)\| \rho) = \log\left({\rm Tr}\left(\sigma^2 \rho^{-1} \right)\right).
\end{equation}
Next using Hadamard's lemma which states that $e^{A} B e^{-A} = \sum_{p=0}^\infty \frac{{\rm ad}_A^p(B)}{p!}$ where ${\rm ad}_A^p(B)$ is the $p$-fold nested commutator of $B$ with respect to $A$.
\begin{align}
    &\partial_{\theta} D_2(\sigma(\theta)\|\rho) = \frac{{\rm Tr}\left(\partial_\theta\sigma(\theta) \rho^{-1}\sigma(\theta) + \sigma(\theta)\rho^{-1}\partial_\theta\sigma(\theta)  \right)}{{\rm Tr}\left(\sigma^2 \rho^{-1} \right)}= \frac{{\rm Tr}\left(\{\partial_\theta\sigma(\theta),\sigma(\theta)\} \rho^{-1}\right)}{{\rm Tr}\left(\sigma^2 \rho^{-1} \right)}\nonumber\\
    &= \frac{{\rm Tr}\left(\left\{\frac{{\rm Tr_h}\left(\partial_{\theta} e^{-H(\theta)} \right)}{{\rm Tr}\left( e^{-H(\theta)}\right)}-\frac{\sigma(\theta){\rm Tr}\left(\partial_{\theta} e^{-H(\theta)} \right)}{{\rm Tr}\left( e^{-H(\theta)}\right)} ,\sigma(\theta)\right\} \rho^{-1}\right)}{{\rm Tr}\left(\sigma^2 \rho^{-1} \right)}\nonumber\\
    &= \frac{{\rm Tr}\left(\left\{\frac{{\rm Tr_h}\left(\partial_{\theta} e^{-H(\theta)} \right)}{{\rm Tr}\left( e^{-H(\theta)}\right)} ,\sigma(\theta)\right\} \rho^{-1}\right)}{{\rm Tr}\left(\sigma^2 \rho^{-1} \right)}+2\frac{{\rm Tr}\left(\left(\partial_{\theta} H(\theta)\right)e^{-H(\theta)} \right)}{{\rm Tr}\left( e^{-H(\theta)}\right)}\nonumber\\
    &= -\int_0^1\frac{{\rm Tr}\left(\left\{{{\rm Tr_h}\left( e^{-sH(\theta)} (\partial_{\theta} H)e^{-(1-s)H(\theta)} \right)}{} ,\sigma(\theta)\right\} \rho^{-1}\right)}{{\rm Tr}\left(\sigma^2 \rho^{-1} \right){\rm Tr}\left( e^{-H(\theta)}\right)}\mathrm{d}s+2\frac{{\rm Tr}\left(\left(\partial_{\theta} H(\theta)\right)e^{-H(\theta)} \right)}{{\rm Tr}\left( e^{-H(\theta)}\right)}\nonumber\\
    &= -\int_0^1\frac{{\rm Tr}\left({\left( e^{-sH(\theta)} (\partial_{\theta} H)e^{-(1-s)H(\theta)} \right)} (\sigma(\theta)\rho^{-1}\otimes I_h+\rho^{-1}\sigma(\theta)\otimes I_h)\right)}{{\rm Tr}\left(\sigma^2 \rho^{-1} \right){\rm Tr}\left( e^{-H(\theta)}\right)}\mathrm{d}s\nonumber\\
    &\qquad+2\frac{{\rm Tr}\left(\left(\partial_{\theta} H(\theta)\right)e^{-H(\theta)} \right)}{{\rm Tr}\left( e^{-H(\theta)}\right)}\nonumber
    \end{align}
    \begin{align}
    &= -\sum_{p=0}^\infty\int_0^1\frac{{\rm Tr}\left({\rm ad}_{-H}^{p}(\partial_\theta H)\frac{s^p}{p!} e^{-H}(\sigma(\theta)\rho^{-1}+\rho^{-1}\sigma(\theta))\otimes I_h\right)}{{\rm Tr}\left(\sigma^2 \rho^{-1} \right){\rm Tr}\left( e^{-H(\theta)}\right)}\mathrm{d}s+2\frac{{\rm Tr}\left(\left(\partial_{\theta} H(\theta)\right)e^{-H(\theta)} \right)}{{\rm Tr}\left( e^{-H(\theta)}\right)}\nonumber\\
    &= -\sum_{p=0}^\infty\frac{{\rm Tr}\left({\rm ad}_{-H}^{p}(\partial_\theta H) e^{-H}(\sigma(\theta)\rho^{-1}+\rho^{-1}\sigma(\theta))\otimes I_h\right)}{{\rm Tr}\left(\sigma^2 \rho^{-1} \right){\rm Tr}\left( e^{-H(\theta)}\right)(p+1)!}+2\frac{{\rm Tr}\left(\left(\partial_{\theta} H(\theta)\right)e^{-H(\theta)} \right)}{{\rm Tr}\left( e^{-H(\theta)}\right)}\nonumber\\
    &= -\sum_{p=0}^\infty\frac{{\rm Tr}\left({\rm ad}_{-H}^{p}(\partial_\theta H) e^{-H}(\{\sigma(\theta),\rho^{-1}\})\otimes I_h\right)}{{\rm Tr}\left(\sigma^2 \rho^{-1} \right){\rm Tr}\left( e^{-H(\theta)}\right)(p+1)!}+2\frac{{\rm Tr}\left(\left(\partial_{\theta} H(\theta)\right)e^{-H(\theta)} \right)}{{\rm Tr}\left( e^{-H(\theta)}\right)}\label{eq:longDeriv}
\end{align}
Here we use the notation that $I_h$ is the identity on the hidden subsystem and take $\{A,B\} = AB + BA$ to be the anti-commutator.  Now, if we change the training set slightly to define $\rho\mapsto (\rho^+)^{-1}$ where $\rho^+$ is the Moore-Penrose pseudo-inverse of a matrix then

\begin{equation}
    \partial_{\theta} D_2(\sigma(\theta)\|(\rho^+)^{-1})= -\sum_{p=0}^\infty\frac{{\rm Tr}\left({\rm ad}_{-H}^{p}(\partial_\theta H) e^{-H}(\{\sigma(\theta),\rho^{+}\}\otimes I_h)\right)}{{\rm Tr}\left(\sigma^2 \rho^{+} \right){\rm Tr}\left( e^{-H(\theta)}\right)(p+1)!}+2\frac{{\rm Tr}\left(\left(\partial_{\theta} H(\theta)\right)e^{-H(\theta)} \right)}{{\rm Tr}\left( e^{-H(\theta)}\right)}
 \end{equation}
Note that this form is well defined even if $\rho$ is rank-deficient, and the above observations combined provide a proof for Theorem~\ref{thm:boltz}.

\qbmgradient*

\subsection{Query Complexity of Training QBMs}
Here we discuss the query complexity of training QBMs.  In contrast to earlier results which give the gate complexities, we focus here on the query complexity owing to the challenges involved in preparing approximate thermal states on quantum computers, which make them less accessible for near term quantum computers than QBMs.  Of course, the ease with which a correspondence between classical and quantum Boltzmann machines can be drawn along with the fact that the issues of entanglement induced barren plateaus can be directly addressed~\cite{marrero2020entanglement} makes such models appealing both theoretically and for long-term applications.  For information about training with respect to forward divergences see~\cite{wiebe2019generative}.

In this section, we will assume three oracles.  First, we assume Prepare and Select oracles which prepare the coefficients of the Hamiltonian terms and applies a selected term from the Hamiltonian to a quantum state vector.  In addition, we assume another oracle $\mathcal{O}_{e^{H}}$ which prepares the thermal state corresponding to the Hamiltonian.  Note that thermal state preparation is not efficient unless $\QMA \subseteq \BQP$; however, due to the diversity of methods employed in thermal state preparation~\cite{poulin2009sampling,yung2012quantum,van2017quantum,chowdhury2020variational} we leave this operation as an oracle rather than employing a specific method that uses Prepare and Select to construct such a state.

\begin{theorem}
 Let $H=\sum_{j=1}^N \alpha_j U_j$ for unitary $U_j$ and $\alpha_j\ge0$ and let Prepare: $\ket{0} \mapsto \sum_j \sqrt{\alpha_j} \ket{j} /\sqrt{\sum_j \alpha_j}$ and Select: $\ket{j}\ket{\psi} \mapsto \ket{j} U_j \ket{\psi}$ and let $\mathcal{O}_{e^H}$ prepare the thermal state of the Boltzmann machine for a given set of weights ($\theta$).  Then there exists $\Upsilon>0$ such that for all $\epsilon < \Upsilon$ there is a quantum algorithm that can estimate $\partial_\theta D(\sigma_v \| \rho^{-1})$ within error $\epsilon$  with probability at least $2/3$ using $O(N e^{2\sum_j \alpha_j}/\epsilon)$ queries to these oracles.
\end{theorem}
 \begin{proof}
 Our algorithm for estimating the gradient for the reverse divergence uses the generalized swap test to evaluate each term that appears in the expression for the gradient found in Theorem~\ref{thm:boltz}.  We can describe the gradient expression in the following form:
 \begin{align}
     &\frac{\sum_{p,\{j\}}{\rm Tr}(H^{j_1} \cdots H^{j_p} (\partial_{\theta} H) H^{j_{p+1}} \cdots H^{j_{2p}}(e^{-H} / {\rm Tr}(e^{-H}) \{\sigma_v,\rho^{-1} \}\otimes I_h)}{{\rm Tr}(\sigma_v^2 \rho^{-1})(p+1)!} + 2 {\rm Tr}\left(\frac{(\partial_\theta H) e^{-H}}{{\rm Tr}(e^{-H})}\right),\nonumber\\
     &\le \sum_{p,\{j\}}\frac{C_{j,p}}{{\rm Tr}(\sigma_v^2 \rho^{-1})}+ 2 {\rm Tr}\left(\frac{(\partial_\theta H) e^{-H}}{{\rm Tr}(e^{-H})}\right)\label{eq:boltzDeriv}
 \end{align}
 where $\{j\}$ describes the set of all $j$ such that $j_1,\ldots, j_{2p}$ are integers mod $2$ and their sum is equal to $p$.
 As $\mathcal{O}_{e^{-H}}$ is a unitary operation and $\|\partial_\theta H\| = 1$, amplitude amplification can be used to estimate the final term in this sum within error $\epsilon/2$ with probability greater than $5/6$ using $O(1/\epsilon)$ queries according to Theorem~\ref{thm:amp_est}.
 
 Now let us turn our attention to the first terms in the summation in~\eqref{eq:boltzDeriv}.  Note that each $|C_{j,p}| \le 2(2 \sum_j \alpha_j)^{p} {\rm Tr}(\sigma_v^2 \rho^{-1})/(p+1)!$ from the von Neumann trace inequality.  Therefore if we begin by estimating the denominator and numerator separately, as is done in Theorem~\ref{thm:gradient_generic}, then we find that if the error in the numerator and the denominator are $|\delta_2| \le {\rm Tr}(\sigma_v^2 \rho^{-1})/2$ and $|\delta_1| \le |\sum_{p,\{j\}} C_{j,p}|/2$ then
 \begin{align}
     &\left|\frac{\sum_{p,\{j\}} C_{j,p}}{(p+1)!{\rm Tr}(\sigma_v^2 \rho^{-1})}-\frac{\sum_{p,\{j\}} C_{j,p}+\delta_1}{{\rm Tr}(\sigma_v^2 \rho^{-1}) +\delta_2}\right|\le \frac{|\delta_1|+ 3|\delta_2||\sum_{p,\{j\}}C_{j,p}|}{{\rm Tr}(\sigma_v^2\rho^{-1})}.
 \end{align}
 Thus if we wish to ensure that we can estimate this ratio within error $\epsilon/2$ then it suffices to take $|\delta_2| \le \epsilon {\rm Tr}(\sigma_v^2 \rho^{-1}) /|6\sum_{p,\{j\}} C_{j,p}|$ and $|\delta_1| \le \epsilon {\rm Tr}(\sigma_v^2 \rho^{-1})/2$.
 
 The approach that we use to evaluate this is to use the LCU lemma to implement the sum over $p$~\cite{berry2015simulating}.  In order to do so, the correct block of the larger unitary needs to be specified.  All such projectors can be implemented as $P_0 = (I - (I -2P_0))/2$, which is the sum of two unitaries.  Both require no additional query operations and thus this contribution to the costs.  Thus using the extended swap test on the linear combination of unitary matrices, we can estimate the sum by subtracting the two results from $I$ and $(I-2P_0)$.
 
 Next let us define the coefficient sum of the unitaries in the decomposition of $\sum_{p,\{j\}} C_{p,j}$ to be $\sum_{p,\{j\}} |C_{p,j}| \le 2e^{\sum_j \alpha_j}$.  The extended swap test allows us to learn such terms with an experiment with output probability 
 \begin{equation}
     \frac{1}{2} + \frac{\sum_{p,\{j\}} C_{p,j} }{2\sum_{p,\{j\}} |C_{p,j}|}.
 \end{equation}
 Thus if we wish to use the swap test to estimate the numerator then it suffices to estimate $\sum_{p,\{j\}} C_{p,j} /\sum_{p,\{j\}} |C_{p,j}|$ within error $\delta_3$.  We therefore have that
 \begin{equation}
     |\delta_3| \le \frac{\epsilon }{4 e^{2\sum_j \alpha_j}} \le \frac{\epsilon {\rm Tr}(\sigma_v^2 \rho^{-1})}{2 \sum_{p,\{j\}} |C_{p,\{j\}}|}\qquad \Rightarrow \qquad |\delta_1| \le \epsilon {\rm Tr}(\sigma_v^2 \rho^{-1})/2.
 \end{equation}
 The LCU lemma requires at most $O(N)$ queries to Prepare and Select to create the state used in the generalized swap test.  Thus the complexity of estimating the numerator within the required error tolerance with probability at least $11/12$ is in $O(N e^{2 \sum_j \alpha_j} / \epsilon)$.

Similarly, the denominator can be estimated directly using the extended swap test and we similarly find that we can estimate the denominator within the desired error tolerance using phase estimation (with probability at least $11/12$) using $O(N e^{2\sum_{j} \alpha_j} /\epsilon)$ queries.

From the union bound, the probability that any one of these estimation steps fails is at most
\begin{equation}
    P_{\rm fail} \le \frac{1}{12} + \frac{1}{12} + \frac{1}{6} = \frac{1}{3}.
\end{equation}
Therefore the protocol will succeed using promised number of query operations with success probability at least $2/3$.
 \end{proof}

\section{Thermal State Learning}\label{sec:app_thermal_state}
\subsection{Coherent Algorithm}
The performance of our coherent algorithm is claimed in Theorem~\ref{thm:deep_qnn}, which we restate below for convenience and provide a proof of the performance of our algorithm.

\thmdeep*

\begin{proof}[Proof of Theorem~\ref{thm:deep_qnn}]
We will now focus on the case when $\rho$ is a thermal state $e^{-H}/Tr[e^{-H}]$ and we  are given oracular access to the Hamiltonian $H=\sum_{j=1}^L \alpha_j U_j$ where each $U_j$ is both Hermitian and unitary and a unitary QNN with an output $\sigma = {\rm Tr}_h \left[ \prod_{j=1}^{N} e^{-iH_j\theta_j} \ket{0}\!\!\bra{0} \prod_{j=N}^{1} e^{iH_j\theta_j} \right]$. Note that $\norm{\alpha}_1=\sum_{j=1}^L \alpha_l$. In particular, let us assume a unitary and invertible quantum oracle $\mathcal{O}_U:\ket{j}\ket{\psi} \mapsto \ket{j} U_j \ket{\psi}$ for any state $\ket{\psi}$.  Additionally, let us assume a second oracle $\mathcal{O}_{exp}(t)\ket{j}\ket{\psi} = \ket{j}e^{-iH_j t}\ket{\psi}$ where $H_j$s are the elementary Hamiltonians used for constructing a QNN.

The expression for gradient of the R\'enyi divergence~\eqref{gradient_trotter} then takes form
\begin{equation}
    \partial_{\theta_k} D_2(\sigma_v\|\rho) = \frac{-i{\rm Tr}\left(\left\{{\rm Tr_h}({[\widetilde{H}_k, \sigma]}),\sigma_v\right\} \right)}{{\rm Tr}\left(\sigma_v^2 e^H\right)}. \label{eq:term_grad}
\end{equation}

The cost then, quantified by the total number of gates and queries made is given by~Theorem~\ref{thm:gradient_generic} to be in 
\begin{equation}
    N_{ops} \in O\left(\frac{n_v+N}{\epsilon {\rm Tr}(\sigma_v e^H/{\rm Tr}(e^H))} \right)
\end{equation}
The only issue with this estimate is that it presupposes the existence of an oracle $O_{\rho^{-1}}$ which prepares a copy of the inverse density operator.  We will use queries to Prepare and Select to create a block encoding of the density operator.

In order to create a density matrix from Prepare and Select queries, recall that if $H\ket{E_j} = E_j \ket{E_j}$ then for any unitary $U$
\begin{equation}
    {\rm Tr}_1( \frac{1}{\sqrt{2^n}}\sum_j \ket{j}U\ket{j})= \frac{1}{2^n}{\rm Tr}_1( \sum_j \ket{E_j} U\ket{E_j}) = \frac{1}{2^n} \sum_j U \ket{E_j}\!\bra{E_j} U^\dagger.
\end{equation}
Next let us use these results but denote $\ket{E_j;0}:=\ket{E_j}\ket{0}_a$ to be the same states tensored with an ancillary register.  It then follows that if we define $\frac{1}{\sqrt{2^n}}\sum_j \ket{j}\text{Prepare}^\dagger~\text{Select}~\text{Prepare}\ket{j} := \ket{\phi_{in}}$.  Let us take the Taylor series truncated to order $K$ and as before, approximate both the numerator and the denominator as a linear combination of expectation values. If we are willing to tolerate error at most $\delta_2$, we can a truncation cut-off $K=O\left(\frac{\log{\frac{1}{\delta_2}}}{\log{\log{\frac{1}{\delta_2}}}}\right)$ where $\delta_2$ is the error tolerance budget for this operation.  Thus we have $K/\epsilon \in \widetilde{O}(1/\epsilon)$ and so the contribution to the scaling from $K$ disappears from our bounds as it is a sub-dominant logarithmic factor.  
\begin{align}
    &(I\otimes \bra{0}_a){\rm Tr}_1( \frac{1}{\sqrt{2^n}}\sum_j \ket{j}\text{Prepare}^\dagger~\text{Select}~\text{Prepare}\ket{j})(I\otimes \ket{0}_a) \nonumber\\
    &= \frac{1}{s 2^n} \sum_{j}\sum_{k\le K,l_1,\ldots, l_k} {\frac{\alpha_{l_1}\dots\alpha_{l_k}}{2^kk!}} U_{l_1}\cdots U_{l_k} \ket{E_j}\!\bra{E_j}\sum_{\ell,l_1,\ldots, l_\ell} {\frac{\alpha_{l_1}^*\dots\alpha_{l_\ell}^*}{2^\ell \ell!}} U_{l_\ell}^\dagger\cdots U_{l_1}^\dagger\nonumber\\
    &=\frac{1}{s 2^n} e^{E_j}\ket{E_j}\!\bra{E_j} = \frac{e^{H}}{s2^{n}} = \frac{{\rm Tr}(\rho^{-1})}{s 2^n} \frac{\rho^{-1}}{{\rm Tr}(\rho^{-1})}
    \end{align}
    
    Thus the probability of success for the preparation of the state $\rho^{-1} / {\rm Tr}(\rho^{-1})$ is
    \begin{equation}
        {\rm Pr}(\rho^{-1} /{\rm Tr}(\rho^{-1}) |  \ket{\phi_{in}})  \ge \frac{{\rm Tr}(e^H)}{e^{\|\alpha\|_1} 2^n}
    \end{equation}
    Thus by marking the success state using a reflection operator (which requires $O(n)$ gates) the number of queries needed to boost the probability to $\Theta(1)$
    \begin{equation}
        N_{\rm queries}
        \in \widetilde{O}\left( \sqrt{\frac{{e^{\|\alpha\|_1} 2^n}}{{\rm Tr}(e^H)}}\right), 
    \end{equation}
    and the number of gates is in $\widetilde{O}\left((L+n)\sqrt{\frac{{\rm Tr}(e^H)}{e^{\|\alpha\|_1} 2^n}}\right)$ per query to $\mathcal{O}_{\rho^{-1}}$.
    
Next if we consider the overall query complexity for the gradient evaluation is then upper bounded by the number of queries by assuming that each query made to our oracles in Theorem~\ref{thm:gradient_generic} is a query to $\mathcal{O}_{\rho^{-1}}$ which then leads to a query complexity in 
\begin{equation}
    \widetilde{O}\left(\frac{N+\sqrt{\frac{{e^{\|\alpha\|_1} 2^n}}{{\rm Tr}(e^H)}}}{\epsilon {\rm Tr}(\sigma_v^2 e^H/{\rm Tr}(e^H))}\right),
\end{equation}
and a gate complexity in
\begin{equation}
\widetilde{O}\left(\frac{(L+n)\sqrt{\frac{{e^{\|\alpha\|_1} 2^n}}{{\rm Tr}(e^H)}}}{\epsilon {\rm Tr}(\sigma_v^2 e^H/{\rm Tr}(e^H))}\right),
\end{equation}
This concludes our proof.
\end{proof}

Further, circuits that can be explicitly used to identify the expectation values used in the algorithm are given in Figure~\ref{fig:LCU_nominator} and Figure~\ref{fig:thermal_swapamard}.  These circuits directly use the above reasoning in concert with the extended swap test.

 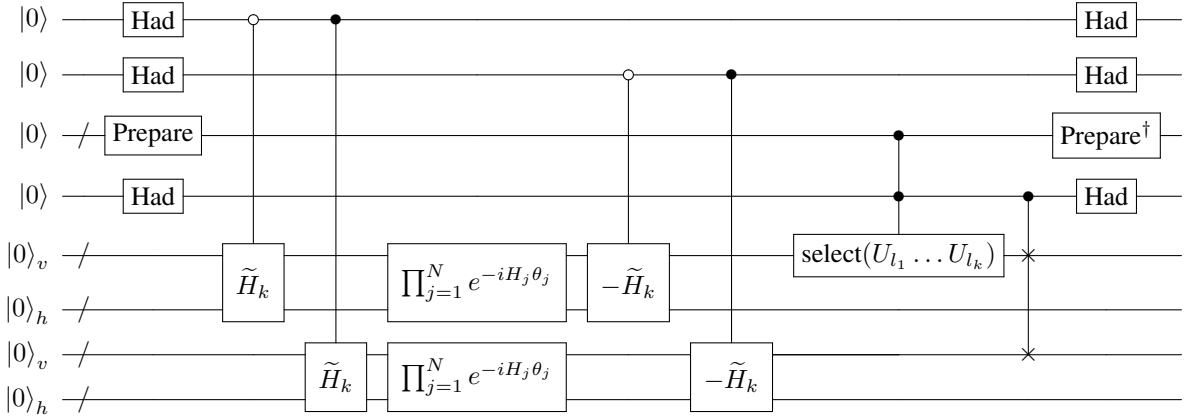
\begin{figure}[t]
 \centering
 $$\Qcircuit @C=0.8em @R=0.8em {
 \lstick{\ket{0}}&\qw&\gate{\text{Had}}&\ctrlo{4}&\ctrl{6}&\qw&\qw&\qw&\qw&\qw&\gate{\text{Had}}&\qw  \\
 \lstick{\ket{0}}&\qw&\gate{\text{Had}}&\qw&\qw&\qw&\ctrlo{3}&\ctrl{5}&\qw&\qw&\gate{\text{Had}}&\qw  \\
 \lstick{\ket{0}}&{/}\qw&\gate{\text{Prepare}}&\qw&\qw&\qw&\qw&\qw&\ctrl{1}&\qw&\gate{\text{Prepare}^{\dagger}}&\qw\\
 \lstick{\ket{0}}&\qw& \gate{\text{Had}}&\qw&\qw&\qw&\qw&\qw&\ctrl{1}&\ctrl{3}&\gate{\text{Had}}&\qw\\
\lstick{\ket{0}_v}&{/}\qw&\qw&\multigate{1}{\widetilde{H}_k}&\qw&\multigate{1}{\prod_{j=1}^{N} e^{-iH_j\theta_j}}&\multigate{1}{-\widetilde{H}_k}&\qw&\gate{\text{select}(U_{l_1}\dots U_{l_k})}&\qswap&\qw&\qw\\
 \lstick{\ket{0}_h}&{/}\qw&\qw&\ghost{\widetilde{H}_k}&\qw&\ghost{\prod_{j=1}^{N} e^{-iH_j\theta_j}}&\ghost{-\widetilde{H}_k}&\qw&\qw&\qw&\qw&\qw\\
 \lstick{\ket{0}_v}&{/}\qw&\qw&\qw&\multigate{1}{\widetilde{H}_k}&\multigate{1}{\prod_{j=1}^{N} e^{-iH_j \theta_j}}&\qw&\multigate{1}{-\widetilde{H}_k}&\qw\qw&\qswap&\qw&\qw\\
 \lstick{\ket{0}_h}&{/}\qw&\qw&\qw&\ghost{\widetilde{H}_k}&\ghost{\prod_{j=1}^{N} e^{-iH_j\theta_j}}&\qw&\ghost{-\widetilde{H}_k}&\qw&\qw&\qw&\qw
 }$$
 \caption{The circuit for state preparation and LCU for the numerator of~\eqref{gradient_thermal}. The gate $\widetilde{H}_k$ can be created using $O(k)$ oracle calls where $k\leq N$. These operations will be repeated multiple time within amplitude estimation.}
 \label{fig:LCU_nominator}
 \end{figure}

%\thmdeep*

\subsection{Sampling based algorithm}\label{sec:shallow}
We first present an algorithm that prioritizes low circuit depth over the overall complexity. In the next section, we will build on this simpler algorithm and present a more efficient algorithm for fault-tolerant quantum computers.

This algorithm estimates the gradient of the numerator and the denominator of~\eqref{gradient_thermal}  by sampling from a sum of weighted estimators. The probability distribution on the weights is easily computable classically and each term in the sum can be estimated by using a quantum circuit depicted in Fig~\ref{fig:thermal_swapamard}.

Specifically, we can estimate
\begin{equation}
    {\rm Tr}(\sigma_v^2 e^H) = {\rm Tr}\left(\sigma_v^2 \sum_{q=0}^{\infty} \frac{(\sum_{j=1}^L \alpha_j H_j)^q}{q!}\right)
\end{equation}
by stochastic sampling. 

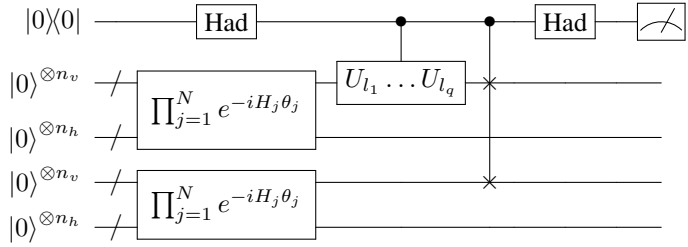
\begin{figure}[t]
    \centering
    $$\Qcircuit @C=0.8em @R=0.8em {
\lstick{\ket{0}\!\!\bra{0}}&\qw&\gate{\text{Had}}&\ctrl{1} &\ctrl{3}&\qw&\gate{\text{Had}}&\qw&\meter\\
\lstick{\ket{0}^{\otimes n_v}}&{/}\qw&\multigate{1}{\prod_{j=1}^{N} e^{-iH_j\theta_j}}&\gate{U_{l_1}\dots U_{l_q}}  &\qswap  &\qw   &\qw&\qw&\qw \\
\lstick{\ket{0}^{\otimes n_h}}&{/}\qw&\ghost{\prod_{j=1}^{N} e^{-iH_j\theta_j}}&\qw&\qw&\qw&\qw&\qw&\qw \\
\lstick{\ket{0}^{\otimes n_v}}&{/}\qw&\multigate{1}{\prod_{j=1}^{N} e^{-iH_j\theta_j}}&\qw &\qswap     &\qw&\qw &\qw&\qw\\
\lstick{\ket{0}^{\otimes n_h}}&{/}\qw&\ghost{\prod_{j=1}^{N} e^{-iH_j\theta_j}}&\qw&\qw&\qw&\qw&\qw&\qw
}$$
    \caption{An example of the extended swap test for sampling terms ${\rm Tr}(\sigma^2U_{l_1}\dots U_{l_q})$. The complexity of this circuit is $O(N + q )$ oracle calls to implement the output of a QNN $\sigma$ and $q$ controlled-$U_{l_i}$ operations. There are $O(n_v)$ gates needed for the multi-qubit  SWAP.}
    \label{fig:thermal_swapamard}
\end{figure}

To do so, we would use the circuit in Fig.~\ref{fig:thermal_swapamard} to sample from terms $Tr(\sigma_v^2 H_{j_1}\dots H_{j_q})$ where the probability of choosing indices $j_1, \dots, j_q$ is proportional to $\frac{\alpha_{j_1}\dots\alpha_{j_1}}{q!}$.
However, the variance of the resulting distribution may be higher than we would like.  We can reduce the variance by using importance sampling.  Let $\vec{j}$ correspond to a vector of length $q$ of the $\alpha_j$ coefficients.  We can then write the result as
\begin{equation}
    {\rm Tr}(\sigma_v^2 e^H) = \sum_{q=0}^{\infty} \sum_{\vec{j}\in\{1,\ldots,N\}^q}  {\rm Tr}(\sigma_v^2\prod_{p\in \vec{j}} H_p ) \frac{(\prod_{p\in \vec{j}} \alpha_p )}{q!}\label{eq:thermal1}
\end{equation}
This reduces the problem to an infinite sum of weighted estimation problems.  Next, we will use importance sampling to write this manifestly as a set of estimation problems wherein the apriori variance of the terms is minimized.  In particular, for any fixed $q$ let 
\begin{equation}
Q(\vec{j}) := \frac{\prod_{p\in \vec{j}} \alpha_j}{\sum_{\vec{j}\in\{1,\ldots,N\}^q} \prod_{p\in \vec{j}}\alpha_p }= \frac{\prod_{p\in \vec{j}} \alpha_p}{(\norm{\alpha}_1)^q}
\end{equation}
We then can write
\begin{equation}
    {\rm Tr}(\sigma_v^2 e^H) = \sum_{q=0}^{\infty}\frac{1}{q!}\left(\sum_{\vec{j}\in\{1,\ldots,N\}^q}\frac{  {\rm Tr}(\sigma_v^2\prod_{p\in \vec{j}} H_p ) {(\prod_{p\in \vec{j}} \alpha_j )}}{{Q(\vec{j})}}\right)Q(\vec{j})
\end{equation}
Similarly, let $s_\ell$ be a random variable such that $\mathbb{E}_{\ell}(s_{\ell,\vec{j}}) = {\rm Tr}(\sigma_v^2\prod_{p\in \vec{j}} H_p )$. Not that we can sample $s_\ell$ from circuit Fig.~\ref{fig:thermal_swapamard}. We then have that 

\begin{align}
    {\rm Tr}(\sigma_v^2 e^H) &= \sum_{q=0}^{\infty}\frac{{(\sum_{p} \alpha_p )}^q}{q!}\left(\sum_{\vec{j}\in\{1,\ldots,N\}^q}{  {\rm Tr}(\sigma_v^2\prod_{p\in \vec{j}} H_p ) }\right)Q(\vec{j})\nonumber\\
    &= \sum_{q=0}^{\infty}\frac{{(\sum_{p} \alpha_p )}^q}{q!}\mathbb{E}_{\vec{j}\ell} (s_{\ell,j}),
\end{align}
where here the expectation value for $s_{\ell,\vec{j}}$ is implicitly taken over the probability distribution $Q(\vec{j})$.  Now let $t_{q}$ be a random variable such that $\mathbb{E}(t_k) = \mathbb{E}_{\vec{j,\ell}} s_{j,\ell}$.  We then have from the additive property of variance that if $\mathbb{V}(t_k) \le \delta_1^2 q!^2 /(\norm{\alpha}_1)^{2k}2^{q+1}$
\begin{equation}
    \mathbb{V}\left( \sum_{q=0}^{\infty}\frac{{(\prod_{p\in \vec{j}} \alpha_j )}^q}{q!}t_k\right)\le \delta_1^2.
\end{equation}
In order to provide an unbiased estimator of $\mathbb{E}(t_k)$ we can simply sample from the extended swap test for randomly chosen indices and shift and rescale the result appropriately.  Each application of the extended swap test~\ref{fig:thermal_swapamard} requires $O(q)$ oracle calls and $O(n)$ gates and we need additional $N$ queries to initialize the extended swap test into states $\sigma_v$.

The query complexity of estimating each $t_k$ within sufficiently small variance is in $O(N+q)$  and has probability $O\left(\frac{ (\norm{\alpha}_1)^{2q}}{\delta_1^2 q!^2}  \right) $.
Therefore, the query complexity of estimating the trace within variance $\delta_1^2$ is in

\begin{align}
    O\left( \sum_{q=0}^\infty \frac{(q +N) (\norm{\alpha}_1)^{2q}}{\delta_1^2 q!^2}  \right) &= O\left(\norm{\alpha}_1\frac{ I_1\left( 2\norm{\alpha}_1 \right)}{\delta_1^2} + N \frac{ (I_0\left( 2 \norm{\alpha}_1 \right)-1)}{\delta_1^2} \right) \\
    &\subseteq O\left(\frac{e^{2\norm{\alpha}_1}}{\delta_1^2 } \left( \sqrt{\norm{\alpha}_1} +  \frac{ N}{\sqrt{2\norm{\alpha}_1} } \right)\right),
\end{align}
where the sum was expressed as a modified Bessel function of the first kind $I_1$.
%I LOVE WOLFRAM ALPHA

The gate complexity of obtaining a single sample is $O(n_v)$. The overall gate complexity is thus
\begin{equation}
        O\left( \sum_{q=0}^\infty \frac{n_v (\norm{\alpha}_1)^{2q}}{\delta_1^2 q!^2}  \right) = 
        O\left(n_v \frac{ (I_0\left( 2 \norm{\alpha}_1 \right)-1)}{\delta_1^2} \right) =  O\left(\frac{  n_v e^{ 2\norm{\alpha}_1 }}{ \delta_1^2\sqrt{2\norm{\alpha}_1} } \right)
\end{equation}

The exact same argument can be applied for estimating ${\rm Tr}\left(\left\{{\rm Tr_h}({[\widetilde{H}_k, \sigma]}),\sigma_v\right\} e^{H}\right)$ with precision $\delta_2$.  To see this, first note that
\begin{align}
    &{\rm Tr}\left(\left\{{\rm Tr_h}({[\widetilde{H}_k, \sigma]}),\sigma_v\right\} e^{H}\right) = {\rm Tr}\left({\rm Tr_h}({[\widetilde{H}_k, \sigma]})\left\{\sigma_v, e^{H}\right\}\right)\nonumber\\
    &={\rm Tr}\left({\rm Tr_h}({\widetilde{H}_k \sigma})\sigma_v e^{H}\right) + {\rm Tr}\left({\rm Tr_h}({\widetilde{H}_k \sigma})e^{H}\sigma_v\right) -{\rm Tr}\left({\rm Tr_h}({\sigma \widetilde{H}_k })\sigma_v e^{H}\right)-
    {\rm Tr}\left({\rm Tr_h}({\sigma \widetilde{H}_k }) e^{H}\sigma_v\right). \label{H_k_position}
\end{align}

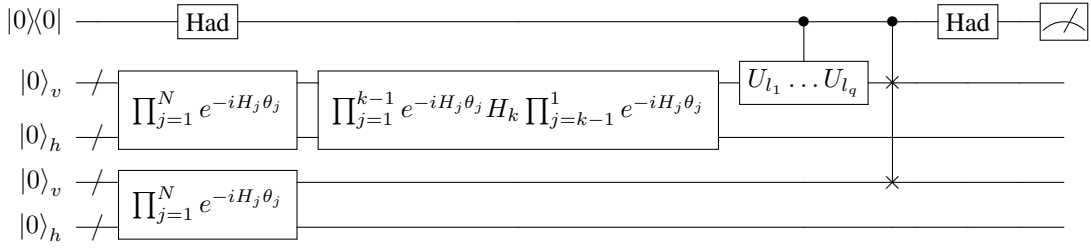
\begin{figure}[t]
    \centering
    $$\Qcircuit @C=0.8em @R=0.8em {
\lstick{\ket{0}\!\!\bra{0}}&\qw&\gate{\text{Had}}&\qw&\ctrl{1} &\ctrl{3}&\qw&\gate{\text{Had}}&\qw&\meter\\
\lstick{\ket{0}_v}&{/}\qw&\multigate{1}{\prod_{j=1}^{N} e^{-iH_j\theta_j}}&\multigate{1}{\prod_{j=1}^{k-1} e^{-iH_j\theta_j} H_k \prod_{j=k-1}^{1} e^{-iH_j\theta_j}}   &\gate{U_{l_1}\dots U_{l_q}}  &\qswap  &\qw   &\qw&\qw&\qw \\
\lstick{\ket{0}_h}&{/}\qw&\ghost{\prod_{j=1}^{N} e^{-iH_j\theta_j}} &\ghost{\prod_{j=1}^{k-1} e^{-iH_j\theta_j} H_k \prod_{j=k-1}^{1} e^{-iH_j\theta_j}}  &\qw&\qw&\qw&\qw&\qw&\qw \\
\lstick{\ket{0}_v}&{/}\qw&\multigate{1}{\prod_{j=1}^{N} e^{-iH_j\theta_j}}&\qw &\qw&\qswap     &\qw&\qw &\qw&\qw\\
\lstick{\ket{0}_h}&{/}\qw&\ghost{\prod_{j=1}^{N} e^{-iH_j\theta_j}}&\qw&\qw&\qw&\qw&\qw&\qw&\qw
}$$
    \caption{An example of the extended swap test for sampling terms from the numerator of~\eqref{gradient_thermal}.
    %${\rm Tr}(\widetilde{H}_k \sigma_v^2U_{l_1}\dots U_{l_q})$. 
    The complexity of this circuit is $O(q + N+k)$ oracle calls to implement state preparation controlled-$U_{l_i}$ and $O(n_v)$ gates needed for the multi-qubit  SWAP.}
    \label{fig:thermal_swapamard_nominator}
\end{figure}

We then see that the result is a sum of four terms each of which is of the same form as~\eqref{eq:thermal1} except for the presence of an additional (unitarily conjugated) Hamiltonian term $\widetilde{H}_k$.  This term can be included in the extended swap test and its cost is additional $O(N)$ calls to oracles $\mathcal{O}_U$ and $\mathcal{O}_{exp}$. Therefore, the query complexity of obtaining a single sample is $O(q+N)$ and following the same calculation as above, the overall query complexity is $ O\left(\frac{e^{2 \norm{\alpha}_1 }}{\delta_2^2 } \left( \sqrt{\norm{\alpha}_1} + \frac{N}{\sqrt{2\norm{\alpha}_1}} \right)\right)$.

Further, an application of Chebyshev's inequality states that (with high probability) the error will be $O(\epsilon)$ if a sufficient number of samples are taken to ensure that this variance bound is met.  It then follows by substitution into~\eqref{eq:bds} that the number of gates required to estimate the gradient within error in the max-norm of $\epsilon$ with probability greater than $2/3$ is in 
\begin{equation}
    O\left(    \frac{n_v e^{2 \norm{\alpha}_1}}{\sqrt{\norm{\alpha}_1} \epsilon^2 {\rm Tr}^2(\sigma_v^2 e^{H} / {\rm Tr}(e^H))} \right).
\end{equation}
Each application of the extended swap circuit requires $O(N+q)$ queries to the oracles that apply each $U_j$ and each $e^{-iH_j \theta_j}$.  Thus the query complexity of gradient estimation is
\begin{equation}
    O\left( \left( \sqrt{\norm{\alpha}_1} + \frac{N}{\sqrt{2\norm{\alpha}_1}} \right)\frac{Ne^{2 \norm{\alpha}_1 }}{\epsilon^2 {\rm Tr}^2(\sigma_v^2 e^{H} / {\rm Tr}(e^H))} \right).
\end{equation}

\begin{restatable}[Sampling-based thermal state learning]{thm}{thmshallow}\label{thm:shallow_qnn}

Let $\rho = e^{-H}/{\rm Tr}(e^{-H})$ be a target distribution for $H=\sum_{l=1}^L \alpha_l U_l$ for unitary $U_l$ and $\sigma = {\rm Tr}_h \left[ \prod_{j=1}^{N} e^{-iH_j\theta_j} \ket{0}\!\!\bra{0} \prod_{j=N}^{1} e^{iH_j\theta_j} \right]$ and output of a unitary QNN on visible units.  For any $\epsilon>0$, the number of gates needed to compute an estimate $\mathcal{E}$ such that $|\mathcal{E} - \partial_{\theta_k} D(\sigma||\rho)| \le \epsilon$ with probability greater than $2/3$ is in 
$$
    O\left( \left( \sqrt{\norm{\alpha}_1} + \frac{N}{\sqrt{2\norm{\alpha}_1}} \right)\frac{Ne^{2 \norm{\alpha}_1 }}{\epsilon^2 {\rm Tr}^2(\sigma_v^2 e^{H} / {\rm Tr}(e^H))} \right).
$$
Similarly, the number of queries  to the oracles $\mathcal{O}_U$ and $\mathcal{O}_{exp}$ needed to perform this computation is in
$$
    O\left(    \frac{n_v e^{2 \norm{\alpha}_1}}{\sqrt{\norm{\alpha}_1} \epsilon^2 {\rm Tr}^2(\sigma_v^2 e^{H} / {\rm Tr}(e^H))} \right).
$$
\end{restatable}
  
\section{Proof of Lemma~\ref{lem:plateau}}\label{sec:plateau}
Here we provide a proof of Lemma~\ref{lem:plateau}, which shows that under broad conditions the mean squared gradient for of the R\'enyi $2$-divergence is not exponentially small under the assumption that the distribution is invariant under unitary transformations.  
\begin{proof}[Proof of lemma~\ref{lem:plateau}]
We begin with our expression for the derivative of the reverse divergence.  First note that ${\rm Tr}(A\otimes B) = {\rm Tr}(A) {\rm Tr}(B)$.  This trick is commonly used to reduce the problem of computing the expectation of a square to an expectation value over multiple copies of a state~\cite{popescu2006entanglement}.  Using this trick,
\begin{align}
    \mathbb{E}((\partial_{\theta_k}\widetilde{D}_2( U\sigma U^\dagger \|\rho))^2)&=\mathbb{E}\left(\frac{{\rm Tr}\left(\left\{(\partial_{\theta_k}U\sigma U^\dagger),U\sigma U^\dagger\right\} \rho^{-1}\right)}{{\rm Tr}\left(U\sigma^{2} U^\dagger \rho^{-1}\right)} \right)^2\nonumber\\
    &=\mathbb{E}\left(\frac{{\rm Tr}\left(\left\{(\partial_{\theta_k}\sigma ),\sigma \right\} U^\dagger \rho^{-1}U\right)}{{\rm Tr}\left(U\sigma^{2} U^\dagger \rho^{-1}\right)} \right)^2\nonumber\\
    &={\rm Tr}\left( \mathbb{E} \left(\frac{U^\dagger\rho^{-1} U }{{{\rm Tr}\left(U\sigma^{2} U^\dagger \rho^{-1}\right)}}\right)^{\otimes 2}\left(\left\{(\partial_{\theta_k}\sigma ),\sigma \right\}  \right)^{\otimes 2}\right)\nonumber\\
    &\ge\frac{{\rm Tr}\left( \mathbb{E} \left(U^\dagger \rho^{-1} U\right)^{\otimes 2}\left(\left\{(\partial_{\theta_k}\sigma ),\sigma \right\}   \right)^{\otimes 2}\right)}{\|\sigma\|^4 {\rm Tr}^{2}(\rho^{-1})} , \label{eq:lastinline}
\end{align}
where the last line follows from the fact that for all $U$, $\rho^{-1}$ is positive definite. 
Next let us assume that $\rho \ket{j} = \gamma_j \ket{j}$ then since the columns of a random unitary matrix are independent up to exponentially small errors,
\begin{align}
    &\mathbb{E}\left((U^\dagger \rho^{-1}  U)\otimes (U^\dagger \rho^{-1}  U)\right)= \sum_{j,j'} \gamma^{-1} \gamma_{j'}^{-1} \mathbb{E}(U^\dagger \ket{j}\!\bra{j} U)\otimes (U^\dagger \ket{j'}\!\bra{j'} U)\nonumber\\
    &=\sum_{j\ne j'} \gamma_j^{-1} \gamma_{j'}^{-1} \mathbb{E}(U^\dagger \ket{j}\!\bra{j} U)\otimes (U^\dagger \ket{j'}\!\bra{j'} U) +\sum_{j} \gamma_j^{-2}  \mathbb{E}(U^\dagger \ket{j}\!\bra{j} U)\otimes (U^\dagger \ket{j}\!\bra{j} U)\nonumber\\
    &=\sum_{j\ne j'} \gamma_j^{-1} \gamma_{j'}^{-1} \frac{(I\otimes I)}{2^{2n}} +\sum_{j} \gamma_j^{-2}  \frac{2(I\otimes I)}{2^n(2^n+1)} + O\left(\frac{{\rm Tr}^2(\rho^{-1})}{2^{3n}} \right)\nonumber\\
    &=\left({\rm Tr}^2(\rho^{-1}) +  {\rm Tr}(\rho^{-2})\right) \frac{(I\otimes I)}{2^{2n}} + O\left(\frac{{\rm Tr}^2(\rho^{-1})}{2^{3n}} \right)\nonumber\\
    &\ge  {\rm Tr}^2(\rho^{-1}) \frac{(I\otimes I)}{2^{2n}} + O\left(\frac{{\rm Tr}^2(\rho^{-1})}{2^{3n}} \right) \label{eq:forwardExpect}
\end{align}
Note that in the third line we have used the fact that $1/(2^n+1) =1/2^n + O(1/2^{2n})$.
In~\eqref{eq:forwardExpect}, the expectation values are computed using the fact that the only components that have a non-zero expectation value are those that have projection onto the symmetric subspace of the two subsystems.  All other components have zero expectation value.  The specific constants come from representation theory as discussed in~\cite{popescu2006entanglement}.  

Eq.~\eqref{eq:forwardExpect} and~\eqref{eq:lastinline} then implies that
\begin{align}
    \frac{{\rm Tr}\left( \mathbb{E} \left(U^\dagger \rho^{-1} U\right)^{\otimes 2}\left(\left\{(\partial_{\theta_k}\sigma ),\sigma \right\}   \right)^{\otimes 2}\right)}{\|\sigma\|^4 {\rm Tr}^{2}(\rho^{-1})} &\ge \frac{{\rm Tr}^2(\left\{(\partial_{\theta_k}\sigma ),\sigma \right\})}{2^{2n}\|\sigma\|^4 }  + O\left(\frac{{\rm Tr}^2(\left\{(\partial_{\theta_k}\sigma ),\sigma \right\})}{2^{3n}\|\sigma\|^4 }\right) \nonumber\\
    &\ge \frac{4{\rm Tr}^2(\sigma(\partial_{\theta_k}\sigma ))}{2^{2n}\|\sigma\|^4 }  + O\left(\frac{{\rm Tr}^2(\left\{(\partial_{\theta_k}\sigma ),\sigma \right\})}{2^{3n}\|\sigma\|^4 }\right) 
\end{align}

Next looking at part $2$ of the lemma, using the forward divergence in~\eqref{eq:forwardDiff}, we have that that for any unitary $U$
\begin{equation}
    \partial_{\theta_k} \widetilde{D}_2(\rho\|U\sigma U^\dagger) = \frac{- {\rm Tr}\left(U^\dagger \rho^2 U \sigma^{-1} (\partial_{\theta_k} \sigma)\sigma^{-1}\right)}{{\rm Tr}\left(U^\dagger \rho^2 U \sigma^{-1}\right)} 
\end{equation}
Therefore following the same reasoning that was used above, 
\begin{align}
    &\mathbb{E}\left(\partial_{\theta_k} \widetilde{D}_2(\rho\|U\sigma U^\dagger) \right)^2 \ge  \sum_{j,j'} \gamma_j \gamma_{j'}{\rm Tr}\left( \mathbb{E}\left(\frac{(U^\dagger \ket{j}\!\bra{j} U) \otimes (U^\dagger \ket{j'}\!\bra{j'} U)}{{\rm Tr}^2\left(U^\dagger \rho^2 U \sigma^{-1}\right)}\right)\left(\sigma^{-1} (\partial_{\theta_k} \sigma)\sigma^{-1})^{\otimes 2}\right)\right)\nonumber\\
    &\ge  \sum_{j,j'} \gamma_j \gamma_{j'}{\rm Tr}\left( \mathbb{E}\left(\frac{(U^\dagger \ket{j}\!\bra{j} U) \otimes (U^\dagger \ket{j'}\!\bra{j'} U)}{{\rm Tr}^2\left(U^\dagger \rho^2 U \sigma^{-1}\right)}\right)\left(\sigma^{-1} (\partial_{\theta_k} \sigma)\sigma^{-1})^{\otimes 2}\right)\right)
\end{align}
Repeating the above reasoning we find the following simple expression
\begin{align}
    \sum_{j,j'} \gamma_j \gamma_{j'}\mathbb{E}\left((U^\dagger \ket{j}\!\bra{j} U) \otimes (U^\dagger \ket{j'}\!\bra{j'} U)\right) &=(1 +{\rm Tr}(\rho^{2})) \frac{(I\otimes I)}{2^{2n}} + O\left(\frac{1}{2^{3n}} \right)
\end{align}
This leads us to the conclusion that
\begin{align}
     &\mathbb{E}\left(\partial_{\theta_k} \widetilde{D}_2(\rho\|U\sigma U^\dagger) \right)^2 \ge \frac{{\rm Tr}^2(\sigma^{-2} (\partial_{\theta_k} \sigma))}{2^{2n} {\rm Tr}^2(\sigma^{-1})} + O\left(\frac{{\rm Tr}^2(\sigma^{-2} (\partial_{\theta_k} \sigma))}{2^{3n} {\rm Tr}^2(\sigma^{-1})} \right)
\end{align}
\end{proof}
In order to understand better the impact that this has on entanglement induced barren plateaus, which occur when there is a volume law for the entanglement entropy between the visible and hidden subsystems of the quantum neural network, let us consider the case where $\sigma\approx I/2^{n_v}$ where $n_v$ is the number of visible units in the system.   Thus we have for the example of the forward divergence
\begin{equation}
    \frac{{\rm Tr}^2(\sigma_v^{-2} (\partial_{\theta_k} \sigma_v))}{2^{2n_v} {\rm Tr}^2(\sigma^{-1})} \lesssim \frac{\|\sigma_v^{-1}\partial_{\theta_k}\sigma_v\|^2}{2^{2n_v}} \lesssim {4\|H_k\|^2\|\sigma_v\|^2} \in O(2^{-2n_v}).
\end{equation}
  It, therefore, holds that the gradient with respect to the visible units of a unitary quantum neural network may still experience entanglement induced barren plateaus.  However, this conclusion does not necessarily hold for derivatives with respect to the hidden weights.  Further work will be needed to determine whether such gradients vanish with high probability over the Haar measure when using the R\'eyni $2$-divergence. 

Repeating the same argument for the reverse divergence leads to
\begin{equation}
    \frac{4{\rm Tr}^2(\sigma(\partial_{\theta_k}\sigma_v ))}{2^{2n_v}\|\sigma_v\|^4 } \lesssim 2^{2n_v+2} {\rm Tr}^2(\sigma_v \partial_{\theta_k} \sigma_v) \lesssim 2^{2n_v+2} \|\partial_{\theta_k} \sigma_v\|^2 \lesssim 16 \|H_k\|^2 \in O(1). 
\end{equation}
Interestingly, this suggests that the lower bound on the derivative of reverse divergence is not necessarily exponentially small.  Thus we anticipate that under some circumstances training with respect to this loss function will not have barren plateaus.  This agrees with our numerical observations that also do not observe a barren plateau for the reverse divergence.  As with the forward divergence, however, more work is needed in order to fully understand the derivatives with respect to the hidden units and give a full treatment of the expectation over the hidden units.

\section{Additional Numerical Experiments}
\label{app:additional_exp}
This section contains a series of additional numerical experiments that expand on the results in Section \ref{sec:thermal} and \ref{sec:ham}.  These results largely explore the role that increasing the number of visible or hidden units have on the performance of the algorithm for very small quantum examples simulated on classical computers.   

\subsection{Larger Model: Quantum Unitary Network}
Using the same experimental formulation outlined in Section \ref{sec:thermal}, we learned an ensemble of higher dimensional thermal states using larger unitary QNNs. We achieved a $98\%$ average fidelity after $100$ epochs for a model with four visible and four hidden units, as shown in Figure \ref{large_unitary_plot}.
\label{large_unitary_exp}
\begin{figure}[htbp]
  \centering%
  \subfloat[Loss]{%
    \includegraphics[width=0.48\textwidth]{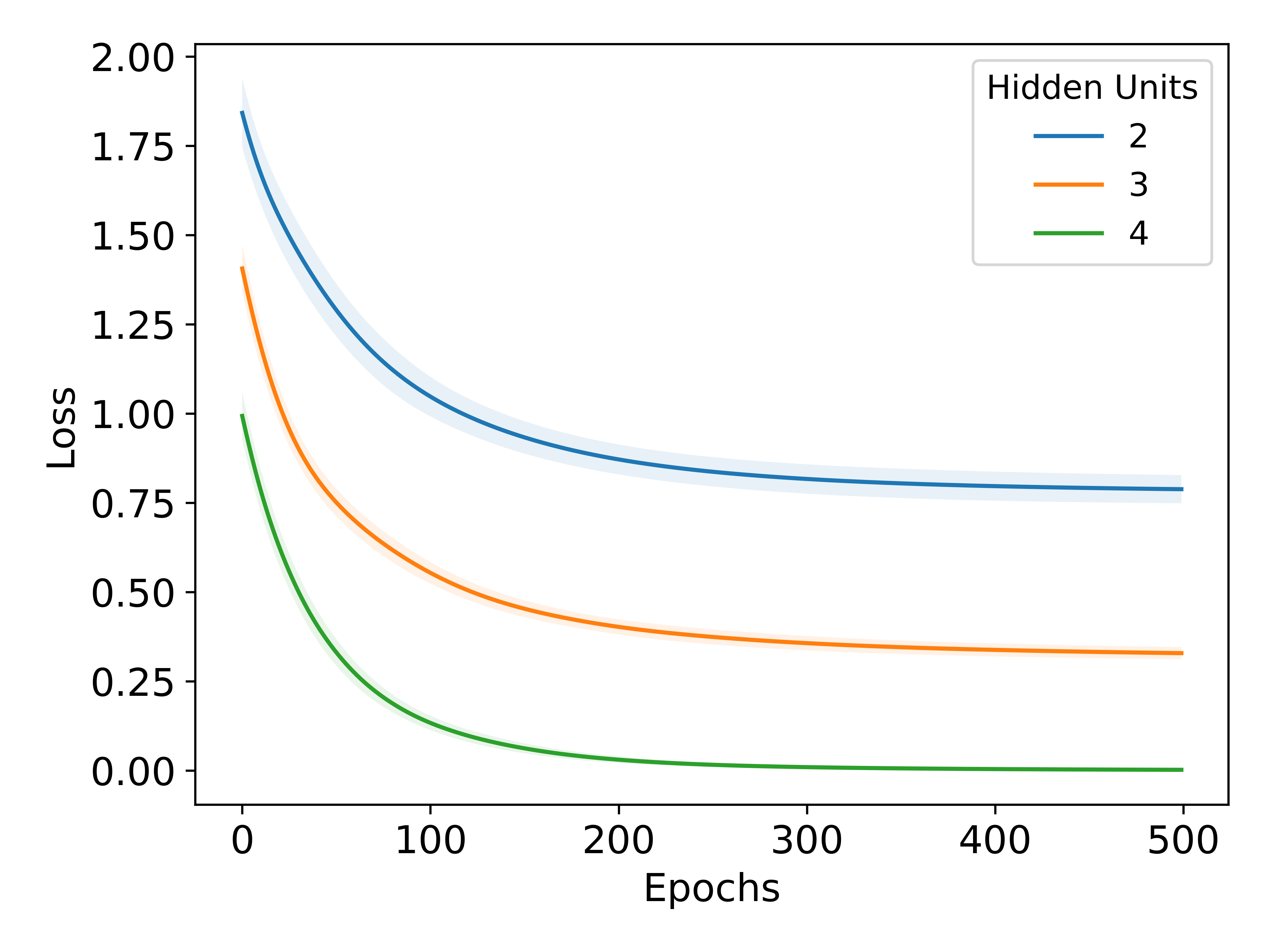}%
  }
  \subfloat[Fidelity]{%
    \includegraphics[width=0.48\textwidth]{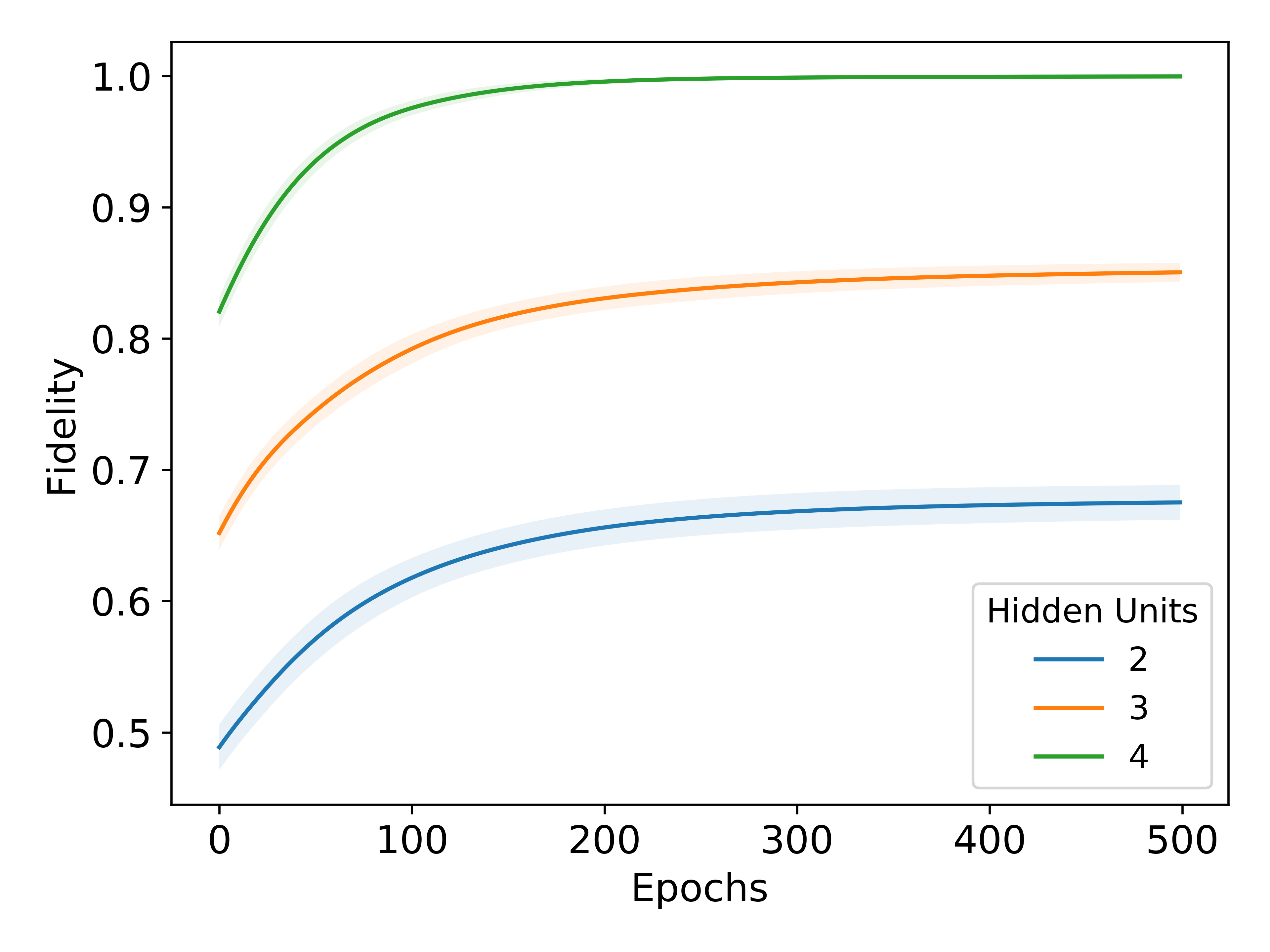}%
  }\\

  \caption{
  Training loss and fidelity of the Unitary model with four visible and an increasing number of hidden units using the ADAM optimizer with an initial learning rate of 0.001. (a-b) We trained a fixed model $50$ different times against different thermal states and computed the mean and standard deviation of the loss per epoch. The solid line represents the mean and the shaded area two standard deviations. The figures show the (a) Training loss and (b) fidelity of the model.}
  \label{large_unitary_plot}

\end{figure}
\FloatBarrier

\subsection{Network Hyperparameters: Quantum Unitary Network}
\label{hyperparamers_exp}
When training a unitary QNN there are two hyperparameters you need to choose from: circuit architecture and network initialization. By circuit architecture we mean the order in which we apply the gates in your quantum circuit because the Unitaries might not commute (see Figure \ref{fig:uqnn} for an example). In Figure \ref{hyperparamers_plot}, we fixed a thermal state and studied the effect that different circuit architectures and network initialization have on training. We initialized our network coefficients by sampling from a normal distribution with mean $0$ and variance $1$. In both experiments, we saw convergence to the same loss value after $300$ epochs. 

\begin{figure}[htbp]
  \centering%
  \subfloat[Architecture ensemble]{%
    \includegraphics[width=0.48\textwidth]{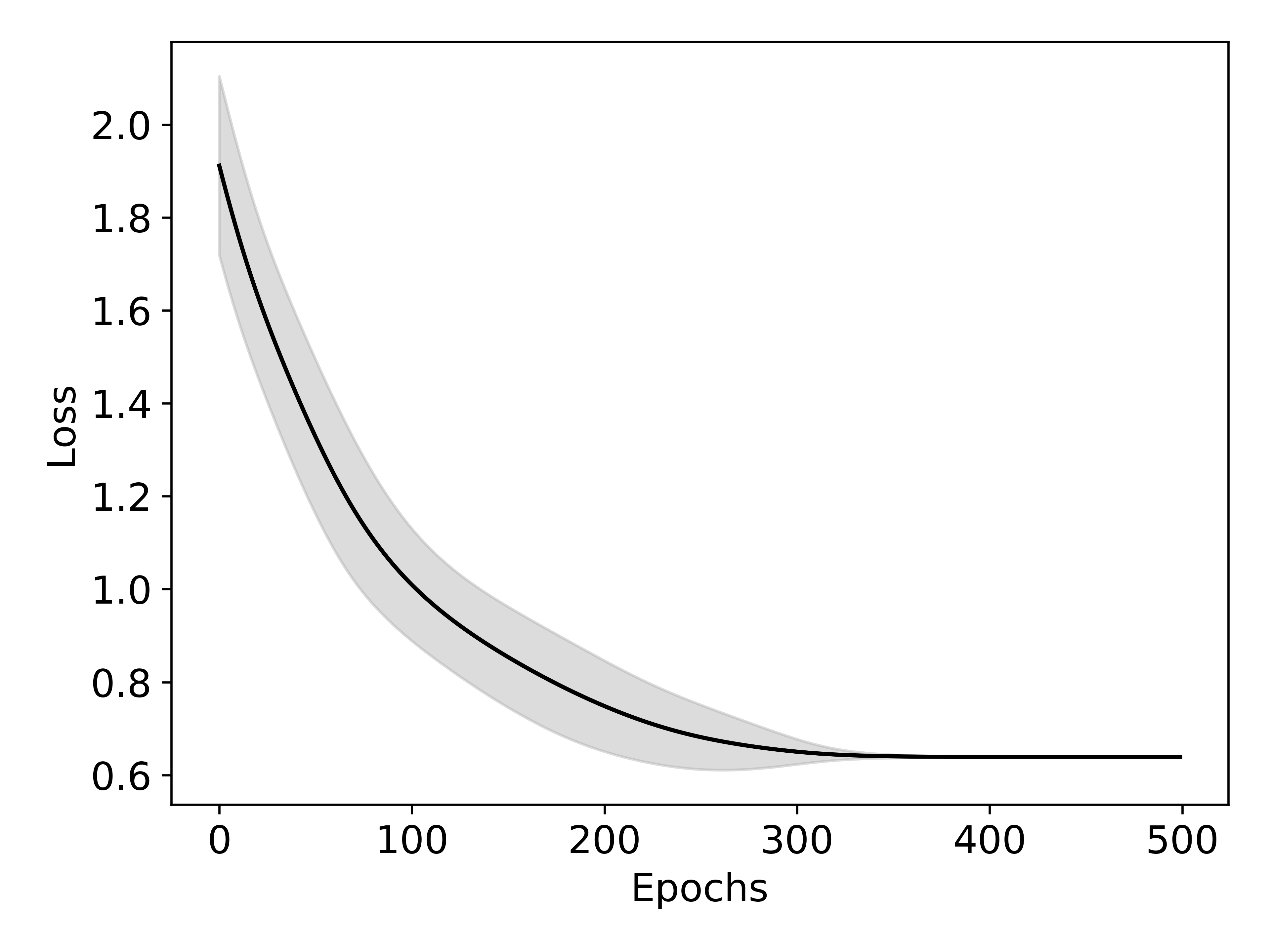}%
  }
  \subfloat[Network ensemble]{%
    \includegraphics[width=0.48\textwidth]{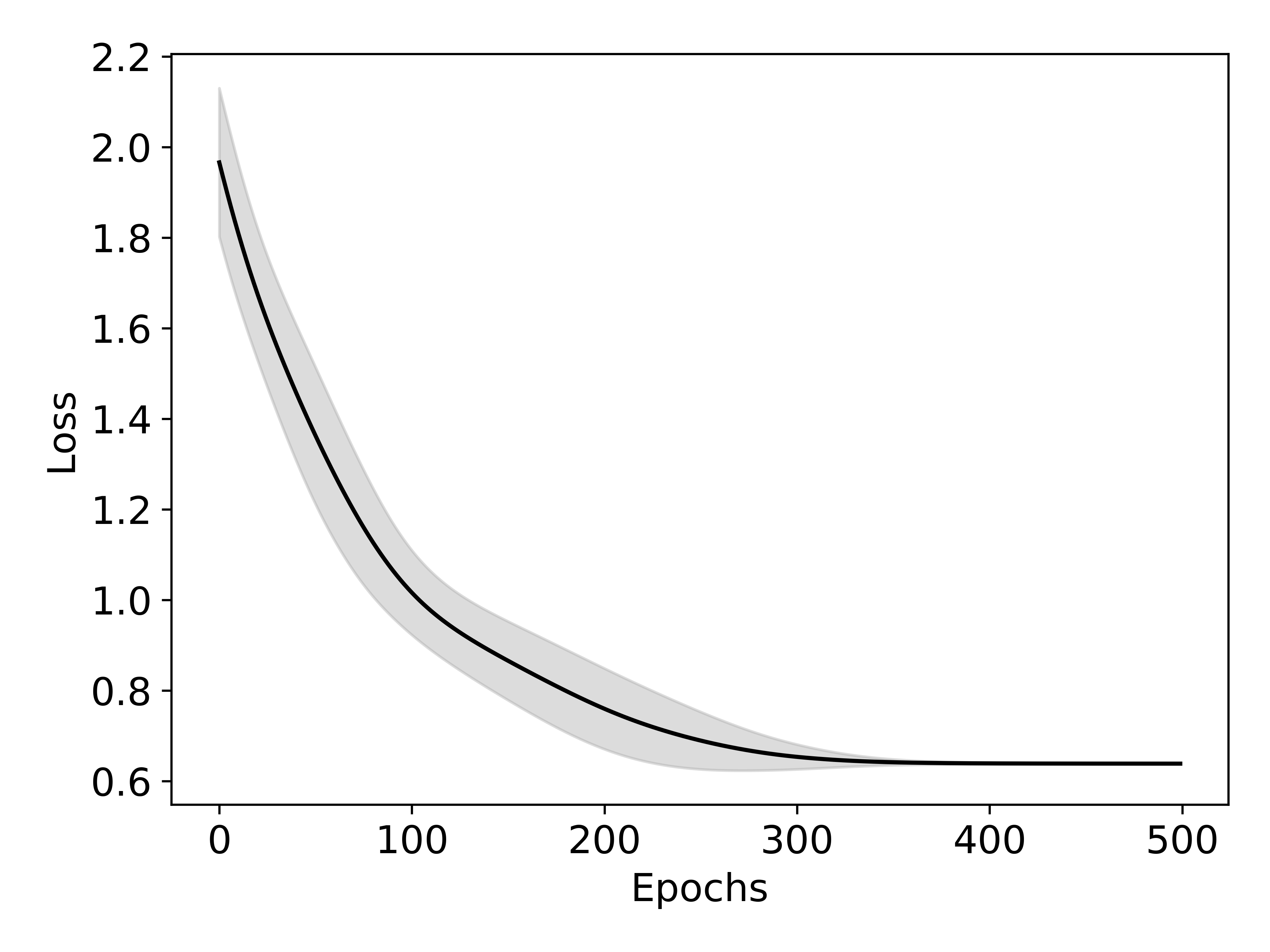}%
  }
 
  \caption{Training loss for the Unitary model with three visible units and one hidden unit using the ADAM optimizer with an initial learning rate of 0.001. (a-b) We trained each model $50$ different times against a fixed thermal state and computed the mean and standard deviation of the loss per epoch. The solid line represents the mean and the shaded area two standard deviations. (a) We varied the order of operations of the ansatzes and (b) the initialization parameters.}
  \label{hyperparamers_plot}
\end{figure}
\FloatBarrier

\subsection{L2 Regularization and Gradients: Quantum Boltzmann Machines}
\label{bm_l2_grad}
In \cite{kieferova2017tomography}, the authors outlined the need for regularization when training QMBs to prevent exploding gradients. We found a similar need when training QBMs with the R\'enyi divergence. In Figure \ref{bm_l2_grad_plot}a, we see that unless we regularize the $\infty$-norm of the gradients begin to diverge after $100$ epochs. After using an L2-regularization with a penalty of $2.0$, the problems disappear, as shown in Figure \ref{bm_l2_grad_plot}b.

\begin{figure}[htbp]
  \centering%
  \subfloat[L2 Regularization]{%
    \includegraphics[width=0.48\textwidth]{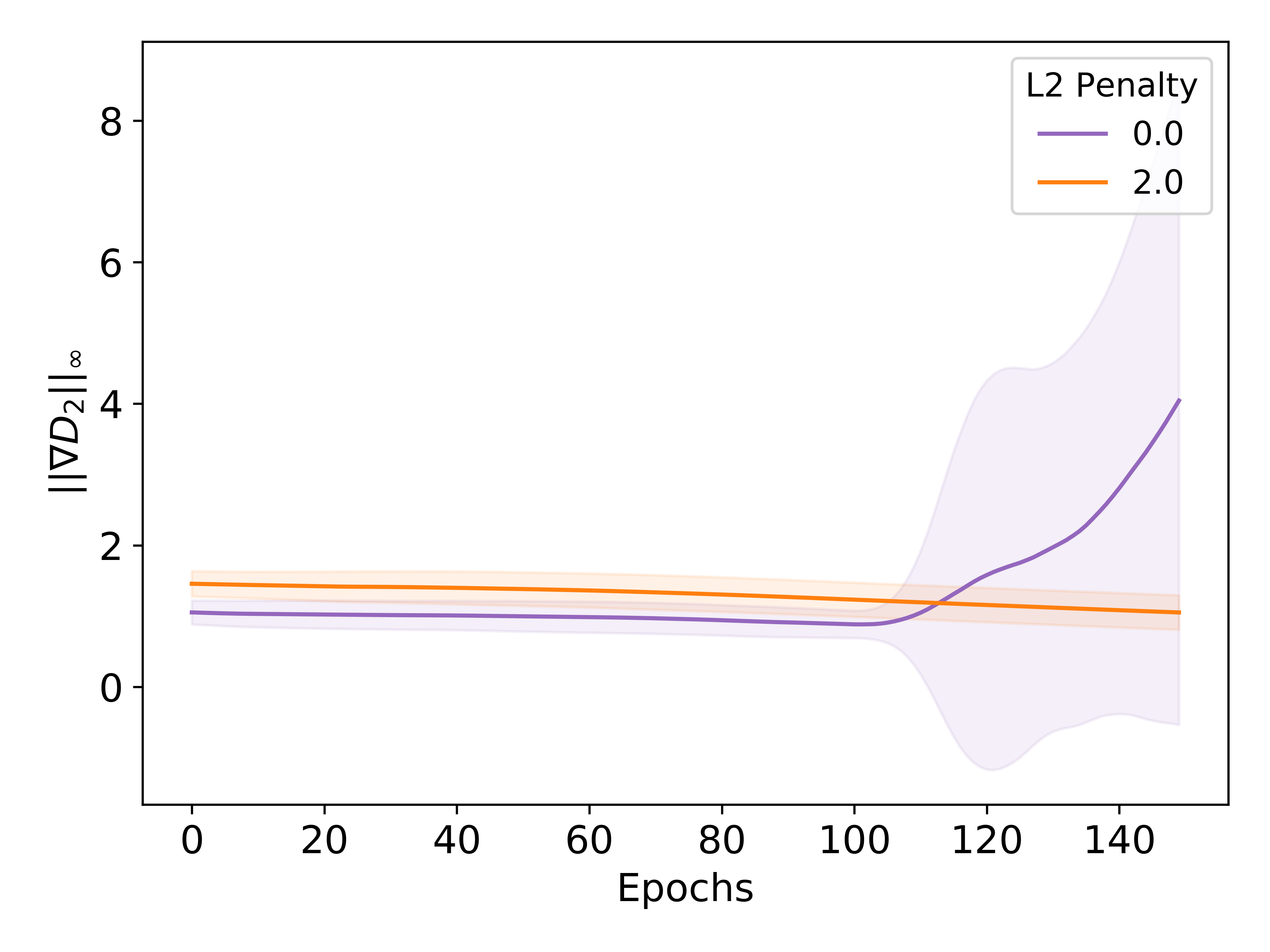}%
  }%
    \subfloat[Hidden Units]{%
    \includegraphics[width=0.48\textwidth]{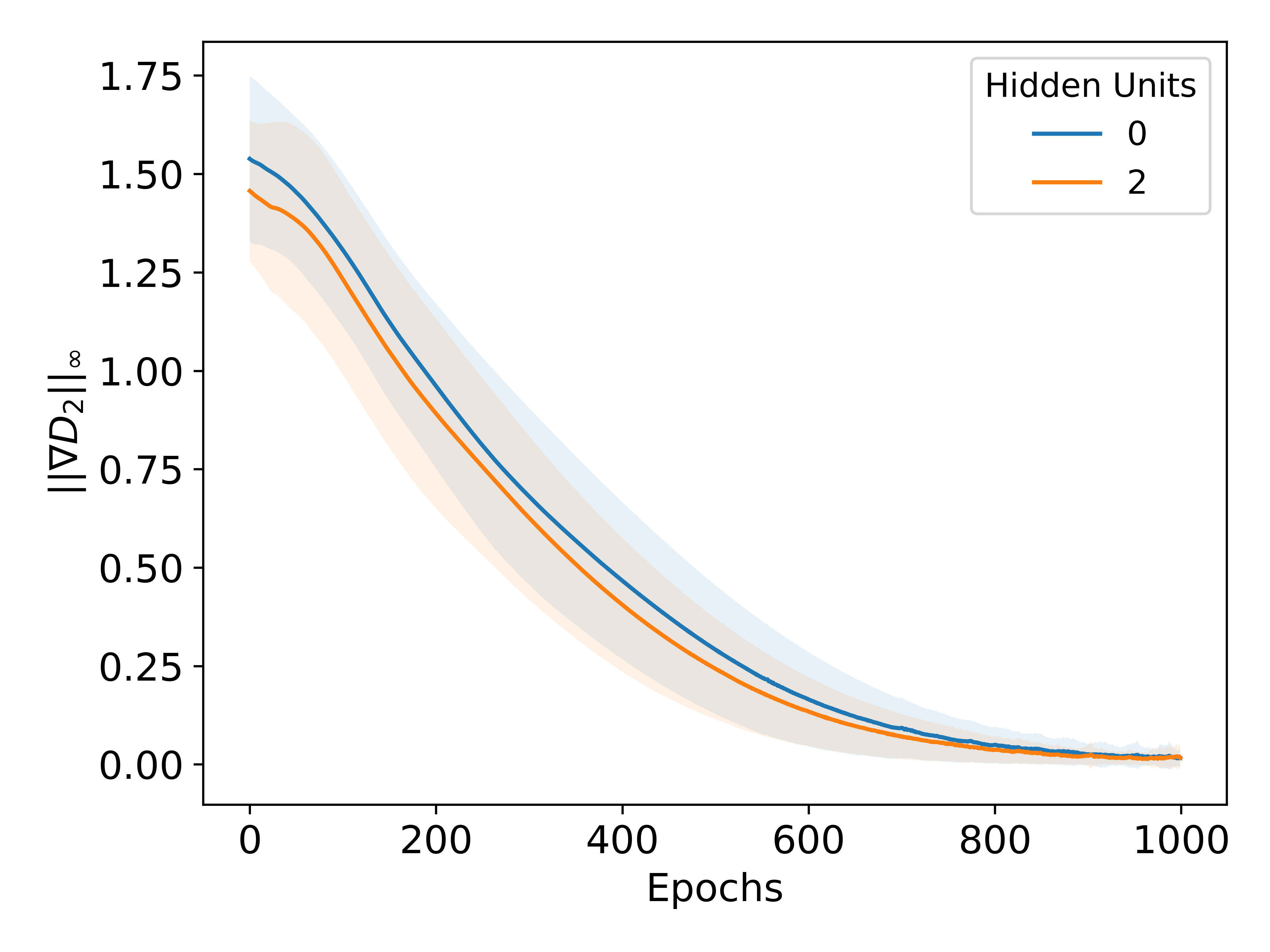}%
  }
  \caption{We monitored the $\infty$-norm of the gradients as we trained Quantum Boltzmann Machines with four visible units. The target states are random thermal states using a three-local Hamiltonian with $\tau=10$. The solid lines represent the average epoch value and the width of the shaded area two standard deviations over an ensemble of 50 different target states. (a) The effect of L2 regularization on the $\infty$-norm of the gradients of a Quantum Boltzmann Machine with two hidden units. (b) Evolution of the $\infty$-norm of the gradients our model in relation to the number of hidden units.}
  
  \label{bm_l2_grad_plot}
 \end{figure}
\FloatBarrier

\subsection{Quantum Boltzmann Machine: High-Temperature State}
\label{bm_high_temp}
Using the same experimental formulation outlined in Section~\ref{sec:ham}, we explored our ability to learn an ensemble of higher temperature states using QBMs. For this class of states, we only managed to achieve a $66\%$ average fidelity after $200$ epochs for a model with four visible and zero or two hidden units, with a starting fidelity of $61\%$ as shown in Figure \ref{bm_high_temp_plot}. Although we have the ability to learn these class of states, we do not see a significant increase in fidelity. This can be the QBMs reaching their approximation capacity.  

\begin{figure}[htbp]
  \centering%
  \subfloat[Loss]{%
    \includegraphics[width=0.48\textwidth]{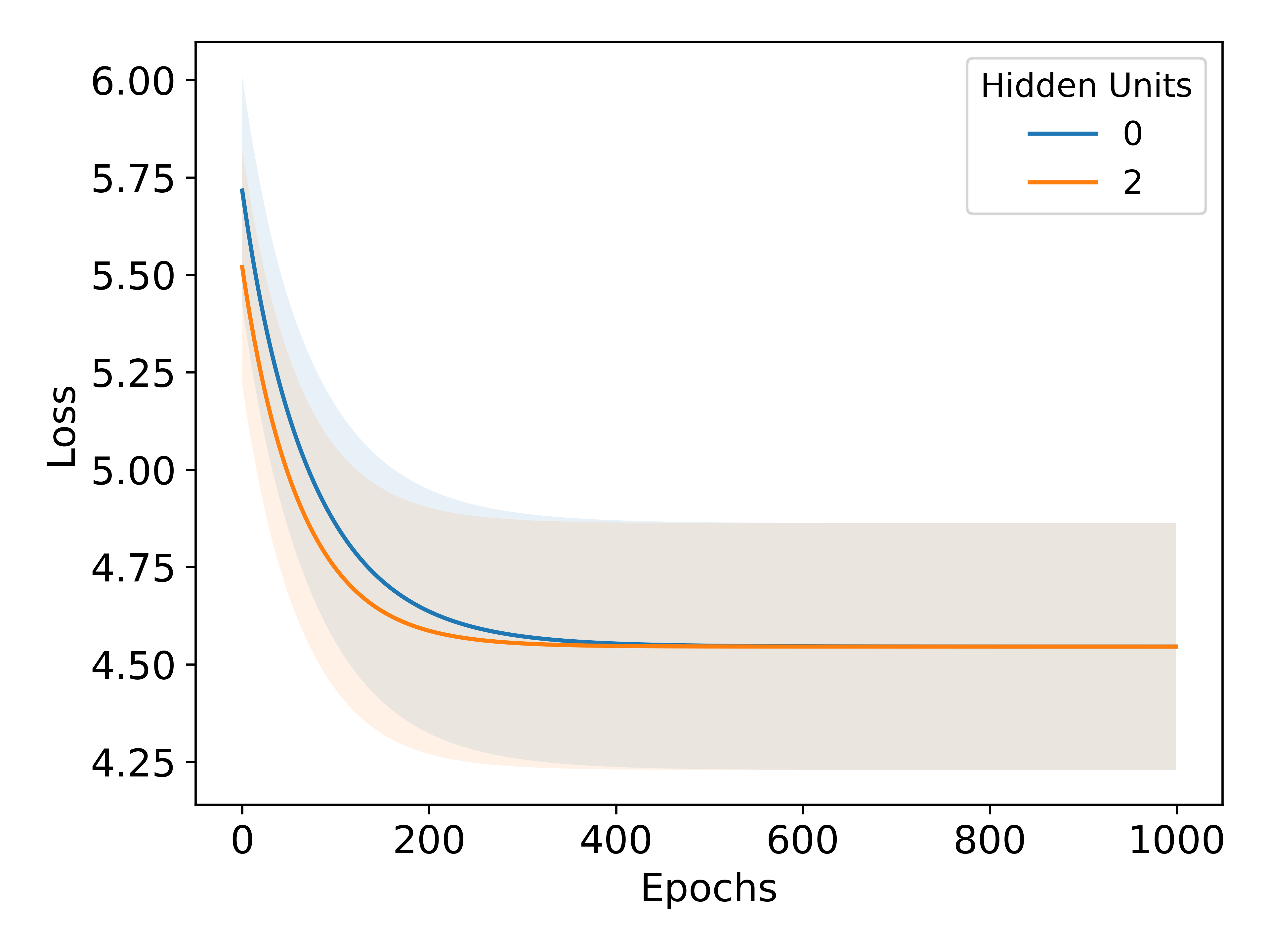}%
  }%
    \subfloat[Fidelity]{%
    \includegraphics[width=0.48\textwidth]{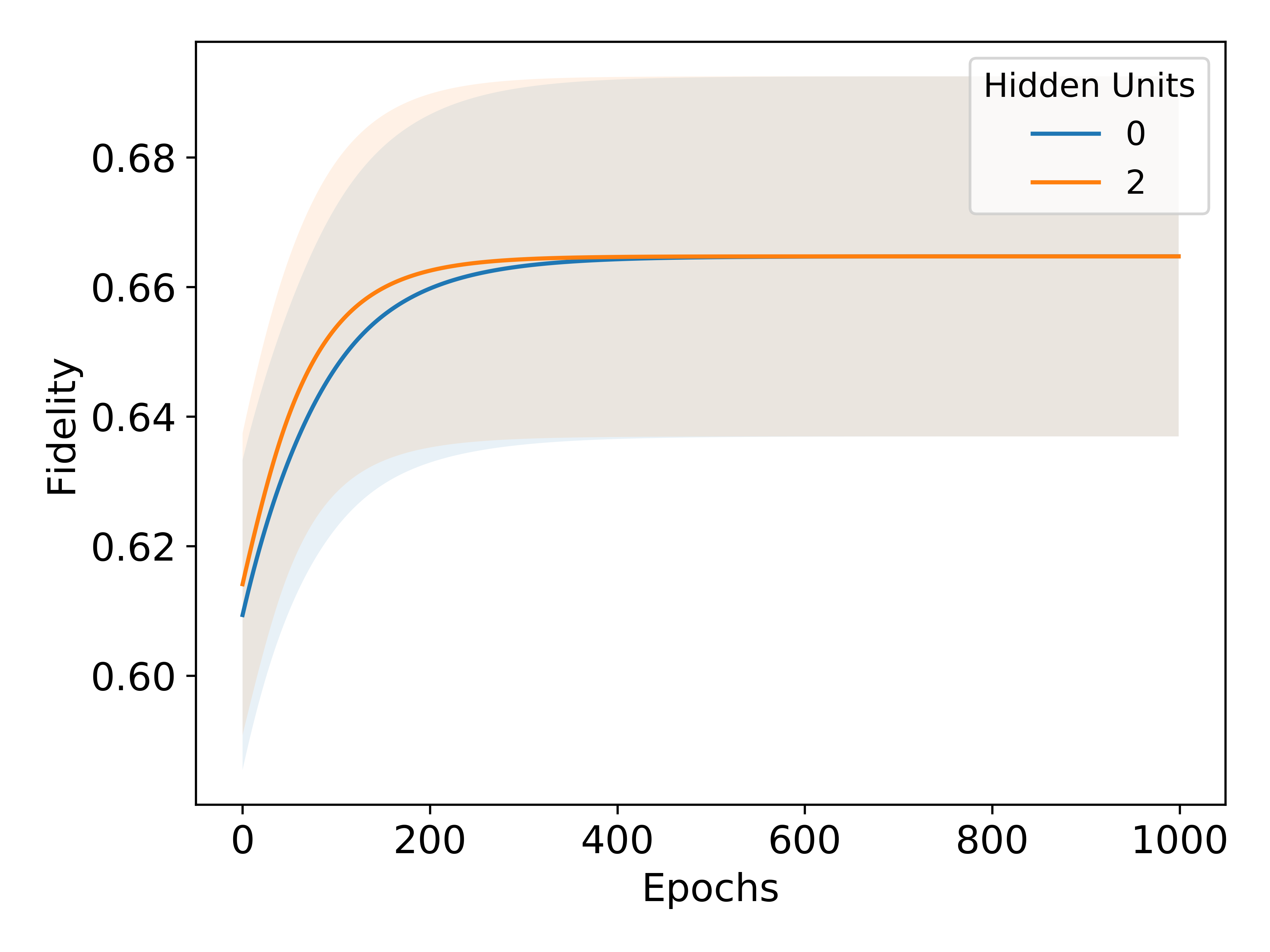}%
  }
  \caption{We trained the Quantum Boltzmann Machine with four visible units and an increasing number of hidden units. The target states are random thermal states using a three-local Hamiltonian with $\tau=5$. The solid lines represent the average epoch value and the width of the shaded area two standard deviations over an ensemble of 50 different target states. (a) Training loss (i.e. R\'enyi Divergence) and (b) Fidelity between the target states and our model.}
  \label{bm_high_temp_plot}
\end{figure}
\FloatBarrier

\end{document}